\documentclass[12pt]{article}
\usepackage{fullpage,epsfig,graphics,amsbsy,amssymb,cancel,bbm,cite}
\usepackage{psfrag,hyperref,color,slashed}
\usepackage{graphicx}
\usepackage{amsfonts,mathdots,dsfont}
\usepackage{amssymb,amsmath,amscd,mathrsfs}
\usepackage{tikz}
\usetikzlibrary{arrows,chains,matrix,positioning,scopes,snakes,fadings}

\makeatletter
\tikzset{join/.code=\tikzset{after node path={%
\ifx\tikzchainprevious\pgfutil@empty\else(\tikzchainprevious)%
edge[every join]#1(\tikzchaincurrent)\fi}}}
\makeatother

\tikzset{>=stealth',every on chain/.append style={join},
         every join/.style={->}}
\tikzset{
    >=stealth',
    punkt/.style={
           rectangle,
           rounded corners,
           draw=black, very thick,
           text width=6.5em,
           minimum height=2em,
           text centered},
    pil/.style={
           ->,
           thick,
           shorten <=2pt,
           shorten >=2pt,}
}

\newcommand{\BB}{\mathbb}

\newcommand{\FR}{\mathfrak}

\def\cA{\mathscr{A}}

\newcommand{\bea}{\begin{eqnarray}}
\newcommand{\eea}{\end{eqnarray}}
\newcommand{\be}{\begin{equation}}
\newcommand{\ee}{\end{equation}}
\newcommand{\nn}{\nonumber}
\newcommand{\Tr}{\textrm{Tr}}

\newcommand{\sbullet}{\vcenter{\hbox{\tiny$\bullet$}}}

\newcommand{\bra}{\langle}
\newcommand{\ket}{\rangle}

\newcommand{\sgn}{\textrm{sgn}}
\newcommand{\reeb}{\vcenter{\hbox{\scriptsize$R$}}}
\newcommand{\sreeb}{\vcenter{\hbox{\tiny$R$}}}
\newcommand{\gYM}{g_{\textrm{\tiny{$YM$}}}}
\newcommand{\gYMf}{g^{\textrm{\tiny{$5D$}}}_{\textrm{\tiny{$YM$}}}}
\newcommand{\gYMt}{g^2_{\textrm{\tiny{$YM$}}}}

\newcommand{\spinc}{\operatorname{spin^c}}

\def\ga{\alpha}
\def\gb{\beta}

\def\Gc{\Gamma}
\def\gd{\delta}
\def\Gd{\Delta}
\def\ep{\epsilon}

\def\gs{\sigma}
\def\ep{\epsilon}

\def\gk{\kappa}
\def\gl{\lambda}

\def\Go{\Omega}
\def\go{\omega}

\def\Gu{\Upsilon}

\DeclareMathAlphabet{\mathpzc}{OT1}{pzc}{m}{it}

\usepackage{amsthm}
\usepackage{amsmath}
\usepackage{amsfonts}
\usepackage{amssymb}

\newtheorem{theorem}{Theorem}[section]

\newtheorem{proposition}[theorem]{Proposition}

\theoremstyle{definition}
\newtheorem{example}[theorem]{Example}
\newtheorem{remark}[theorem]{Remark}

\begin{document}
\begin{flushright} \small
UUITP-30/16
 \end{flushright}
\smallskip
\begin{center} \Large
{\bf $\mathcal{N}=2$ supersymmetric gauge theory \\
on connected sums of $S^2\times S^2$}
 \\[12mm] \normalsize
{\bf Guido Festuccia${}^a$, Jian Qiu${}^{a,b}$, Jacob Winding${}^a$, Maxim Zabzine${}^a$} \\[8mm]
 {\small\it
   ${}^a$Department of Physics and Astronomy,
     Uppsala University,\\
     Box 516,
     SE-75120 Uppsala,
     Sweden\\
   \vspace{.5cm}
      ${}^b$ Mathematics Institute,  Uppsala University, \\
   Box 480, SE-75106 Uppsala, Sweden\\
   }
\end{center}
\vspace{7mm}
\begin{abstract}
We construct 4D $\mathcal{N}=2$ theories on an infinite family of 4D toric manifolds with the topology of connected sums of $S^2 \times S^2$.
 These theories are constructed through the dimensional reduction along a non-trivial $U(1)$-fiber
 of  5D theories on toric Sasaki-Einstein manifolds. We discuss the conditions under which such reductions can be carried out and give a partial classification result of the resulting 4D manifolds. We calculate the partition functions of these 4D theories and they involve both instanton and anti-instanton contributions, thus generalizing  Pestun's famous result on $S^4$.
 \noindent
\end{abstract}

\eject
\normalsize

\tableofcontents
\section{Introduction}

Starting from the work \cite{Pestun:2007rz} there has been huge activity on  studying  supersymmetric theories on curved manifolds and on the exact calculation of their
 partition functions using  localization techniques. The original work  \cite{Pestun:2007rz} was  devoted to $\mathcal{N}=2$ gauge theory on $S^4$, but since then there has been significant progress
  in diverse dimensions (from 2D to 7D) and on diverse backgrounds. For a recent overview of the field see  \cite{Pestun:2016zxk}; localization computations in different dimensions are reviewed in \cite{Pestun:2016jze} (for the 4D case see also \cite{Hosomichi:2016flq}).

We have a precise classification of the geometries on which 4D $\mathcal{N}=1$ theories can be placed preserving supersymmetry (see e.g. \cite{Klare:2012gn, Dumitrescu:2012ha}). The same is true for $\mathcal{N}=2$ in 3D \cite{Klare:2012gn, Closset:2012ru} and $\mathcal{N}=(2,2)$ theories in 2D \cite {Closset:2014pda}. Many  localization calculations have been performed explicitly in lower dimension (2D and 3D) while in four dimensions applications of this technique to $\mathcal{N}=1$ have concentrated on a limited set of geometries \cite{Benini:2011nc, Assel:2014paa, Nishioka:2014zpa, Closset:2013sxa}.
In the case of 4D $\mathcal{N}=2$ theories the situation is even less satisfactory as we do not  yet have a complete classification of the corresponding supersymmetric geometries. In particular, with a view towards applying localization techniques, we are interested in 4D manifolds that admit a toric action.  It is interesting to notice that in 5D there exists a rich class of toric Sasaki-Einstein manifolds that admit $\mathcal{N}=1$ theories. The goal of the present paper is to generate a rich class of toric 4D backgrounds which admit $\mathcal{N}=2$ theories from dimensionally reducing these 5D examples. Essentially we will perform the reduction along non-trivial $U(1)$ fibration of the toric Sasaki-Einstein manifold in order to get a 4D supersymmetric theory. We also derive the exact 4D partition function for these theories.  The manifolds we will consider have topological type $\#_k (S^2\times S^2)$, and are a sub-class of the possible homeomorphism types of smooth simply connected spin 4-manifolds $(\pm M_{E_8})^{\#2m}\#(S^2\times S^2)^{\#k}$.

Using the rigid supergravity approach \cite{Festuccia:2011ws} it is not easy to completely classify the  geometries on which
 4D $\mathcal{N}=2$ theories can be placed preserving supersymmetry (see \cite{Gupta:2012cy, Klare:2013dka, Hama:2012bg, Pestun:2014mja, Butter:2015tra} for progress in this direction).
 The best studied cases are the round sphere \cite{Pestun:2007rz} and the squashed sphere \cite{Hama:2012bg, Pestun:2014mja}.
The squashed sphere can be further generalized to local $T^2$-bundle fibrations \cite{Pestun:2014mja}.
Equivariantly twisted theories on toric K\"ahler surfaces were also considered, with emphasis on $S^2 \times S^2$  \cite{Bawane:2014uka}
and $\mathbb{C}P^2$ \cite{Rodriguez-Gomez:2014eza, Bershtein:2015xfa}.  The study of $\mathcal{N}=2$ theories on $S^2\times S^2$ was also started in \cite{Sinamuli:2014lma}.

The main result of this work  is the explicit construction of  $\mathcal{N}=2$ SYM theories on an infinite family of 4D toric manifolds
 with the topology of connected sums $\#_k (S^2 \times S^2)$ via dimensional reduction from 5D. We would like to stress that our 4D examples are not generically K\"ahler
  and here by toric 4D manifolds we mean 4D manifold with smooth $T^2$-action with the orbit space being convex polytope.
 We start by considering toric Sasaki-Einstein manifolds which admit a free $U(1)$-action that preserves the Killing spinors, and we perform the reduction along this $U(1)$.
  We provide a partial classification of such toric Sasaki-Einstein manifolds.  The resulting 4D theory has  unusual properties originating from the fact that the $U(1)$-fibre does not have a constant size with respect to the Sasaki-Einstein metric.
As a result the 4D theory has a position dependent Yang-Mills coupling. If we add a $\theta$-term to the 4D theory we can introduce the point dependent complex coupling $\tau$, which takes value in the upper half plane
 \bea
  \tau (x) = \frac{4\pi i}{g_{YM}^2 (x)} + \frac{\theta}{2\pi}~,
 \eea
 where $g_{YM}  (x)$ is the 4D dimensionless Yang-Mills coupling  and its dependence from $x$
  comes from the Sasaki-Einstein metric in 5D, see section \ref{sec_Red} for further explanation.
  The connected sum $\#_k (S^2 \times S^2)$ is a toric manifold with $T^2$-action, and it has $(2+2k)$-fixed points.
 The exact  partition functions for these 4D theories is given by the classical term, one-loop term and the instanton term
\bea
	Z = \int\limits_{\FR{t}}da~e^{- \sum\limits_{i=1}^{2k+2} \frac{4\pi^2 r^2}{\epsilon_1^i \epsilon_2^i g^2_{YM}(x_i)} \Tr[a^2]}\cdotp
\frac{{\det}_{adj}' ~  \Gu^C (ia|\reeb^1,\reeb^2)} {\det_{\underline{R}}\Gu^C ( ia+im+\vec\xi\cdotp\vec\reeb/2| \reeb^1,\reeb^2) }  Z_{\mathrm{inst} } ( a | \vec \reeb )  ~,
\eea
where $\vec\reeb$ is related to the $T^2$-action and $\Gu^C$ is a special function which gives the one-loop determinant.
 The above partition function corresponds to the $\mathcal{N}=2$ vector multiplet coupled to a hypermultiplet in representation $\underline{R}$.
 The instanton contributions come from  point-like instantons and anti-instantons which sit on the fixed points $x_i$,
\bea
Z_{\mathrm{inst} } ( a | \vec \reeb )  = \prod_{i=1}^{k +1} Z_{\mathrm{inst}}^{\mathbb{C}^2 } ( a , q_i | \epsilon_1^i, \epsilon_2^i ) \times \prod_{i=k+2}^{2+2k} Z_{\mathrm{inst}}^{\mathbb{C}^2 } ( a , \bar q_i | \epsilon_1^i, \epsilon_2^i )~,
\eea
where
\bea
q_i = q(x_i) = e^{2\pi i \tau(x_i)}~.
\eea
 Here   $Z_{\mathrm{inst}}^{\mathbb{C}^2 } ( a , q_i | \epsilon_1^i, \epsilon_2^i )$  is the Nekrasov partition function on $\mathbb{C}^2$ with equivariant parameters $\vec \epsilon_i$, that can be read off from the fixed points $x_i$.
  Note that the  theories considered here are not the topologically twisted Donaldson-Witten theory, since we have  a mixture of instanton- and anti-instanton-contributions.
  It is possible to specify further the toric geometry and find situations when the instanton and anti-instanton contributions pair together,
   \bea
   Z_{\mathrm{inst} } ( a | \vec \reeb )  = \prod_{i=1}^{k+1 } |Z_{\mathrm{inst}}^{\mathbb{C}^2 } ( a , q_i | \epsilon_1^i, \epsilon_2^i ) |^2~.
  \eea
  Thus our result generalizes Pestun's famous result on $S^4$ \cite{Pestun:2007rz}.

One may get nervous from the fact that $\tau$ depends on $x$. However this is not so exotic and it was discussed  previously  in \cite{Nekrasov:2002qd, Losev:2003py} in the context 
equivariant localization of gauge theories on $\mathbb{R}^4$.
Moreover the gauge theories with $\tau(x)$  can be obtained from the reduction of $(2,0)$ 6D theory on elliptically fibered K\"ahler manifolds \cite{Martucci:2014ema,Assel:2016wcr}. 
 Nevertheless we can deform 5D theory by  performing a Weyl rescaling of our 5D manifold so that the length  of $S^1$-fiber is fixed to  be a constant.
 Through a calculation using the rigid limit of minimal off-shell 5D supergravity, we check that this can be done without breaking supersymmetry. This deformation induces a Q-exact change of the action. After reducing to 4D using the rescaled background, we now find a theory with a constant Yang-Mills coupling, but where the $x$-dependence is now shifted to a $\theta$-term. It is important to stress that the partition function of the theory does not depend on $\tau(x)$ in general, but only on its values at the fixed points.

The paper is organised as follows:
Sections \ref{sec_pdab} and \ref{sec_Acr} are preparatory sections where we analyze the conditions under which the 5D $\mathcal{N}=1$  theory on a
 non-trivial circle fibration can be reduced down to the 4D $\mathcal{N}=2$  theory, while sections \ref{sec_Red} and \ref{sec-partition} contain the main
  result with the explicit construction of 4D $\mathcal{N}=2$  theory and the calculation of its partition function.
    In section \ref{sec_pdab} we discuss in detail the criterion for pushing a bundle down an $S^1$-fibre.
 In particular  the parameters for the supersymmetry transformations are a pair of Killing spinors in 5D, and we seek conditions under which they can be reduced to 4D.
 This allows us to avoid dealing directly with the  supersymmetry algebra in 4D.
 In section \ref{sec_Acr} we specialize to the case of toric Sasaki-Einstein manifolds and we
 present a simple classification of toric Sasaki-Einstein manifolds with a free $U(1)$ isometry preserving the holomorphic volume form (of the Calabi-Yau cone).
 The classification is not that of the regular toric SE manifolds and the resulting 4D geometry, which we study in sections \ref{sec_tgotb} and \ref{sec_IfagoB}, is more interesting.
   With this preparation in section \ref{sec_Red} we reduce the action and the supersymmetry transformations of the 5D supersymmetric gauge theory on an Sasaki-Einstein
    manifold to 4D. We also discuss various features of the reduced $N=2$ 4D theory and consider some of its supersymmetric deformations.
In section \ref{sec-partition} we discuss the partition function of the 4D theories, which can be obtained discarding non-zero Kaluza-Klein modes.
    We also consider  the issue of assembling the instanton sector. Due to the misalignment of the aforementioned freely acting $U(1)$ and the Reeb vector field, one gets a mixture of instantons and anti-instantons. This is a main new feature of our theory that distinguishes it from the Donaldson-Witten theory. The paper contains appendices which complement the main text with some  background and  technical considerations.

\section{Conditions for reduction} \label{sec_pdab}

Performing dimensional reduction is straightforward if the 5D manifold is a trivial $S^1$ bundle over a 4D base manifold.
If the $S^1$ bundle is non-trivial it is still possible to reduce. Locally this is Scherk-Schwarz reduction \cite{Scherk:1979zr} but, since we are considering compact manifolds, we need to identify under which conditions there are no global obstructions.
 We will see that  stating these conditions for differential forms is straightforward, but for spinors the issue is a more subtle.
In general, the various fields that we wish to dimensionally reduce are sections of some vector bundles over our manifold.
Hence, we will consider when bundles and sections of these bundles can be consistently pushed down from the 5D manifold to the 4D base.
In the following we will state the relevant facts and give some examples.
Proofs are presented in appendix \ref{app_proofs}.

To set our notation, let $S^1\to M\to B$ be a nontrivial circle fibration, and $E\to M$ be a vector bundle.
We first give a criterion for being able to push the bundle $E$ down to $B$.
If $E$ possesses a trivialization over patches of the form $[0,2\pi]\times U_i$, with $\{U_i\}$ a cover of the base manifold $B$, such that the transition functions are independent of the circle direction, then $E$ can be pushed down to $B$.
We can reformulate this criterion as follows: Denote the coordinate of the circle fibre as $\ga$ and let $A$ be a connection of $E$, then if $P\exp i\int_0^{2\pi} d\ga A_{\ga}=id$, the bundle $E$ can be pushed down. Moreover, when this is satisfied,  sections of $E$ such that $D_{\ga}s=0$ can be pushed down.

The push down is not unique but depends on the choice of connection. As an example, consider $S^5$ as the total space of the Hopf bundle $S^1\to S^5\stackrel{\pi}{\to}\BB{P}^2$. We want to push down the trivial bundle $S^5\times\BB{C}$ to $\BB{P}^2$. One way is to choose the zero connection and the pushdown is also the trivial bundle. Alternatively one may choose $A=n\gk$, where $\gk$ is the contact 1-form of the contact structure of $S^5$ associated with the Hopf fibration structure. Concretely $\gk$ is described as $d\ga+\pi^*\cA$ where $\cA$ is the connection on $\BB{P}^2$ of the bundle ${\cal O}(1)$. The holonomy of $\gk$ is $2\pi$, and the push down is ${\cal O}(n)$.

The above example shows that pushing down is not a canonical procedure. On the other hand, we know that there exists a canonical procedure to push down differential forms. Denote with $X=\partial_{\ga}$ the vector field along the circle fiber. A 1-form $\xi$ that satisfies $\iota_X\xi=0$ and $L_X\xi=0$  (where $L_X$ is the Lie derivative along $X$ and $\iota_X$ the contraction of a form with $X$)  can be regarded as a 1-form on $B$ canonically.

 The following example clarifies this issue. Consider  the subbundle $T_H^*M$ of 1-forms $\xi$ with $\iota_X\xi=0$, i.e. horizontal 1-forms. We want to push it down to $B$.

 First without any loss of generality, we can choose a metric such that $X$ is Killing and normalized to $\bra X,X\ket=1$.
  It follows that $J_{\mu\nu}=-\nabla_{\mu}X_{\nu}$ is anti-symmetric. From $2 X^{\rho}\nabla_{\rho}X_{\mu}=2 X^{\rho}\nabla_{\mu}X_{\rho}=\partial_{\mu}\bra X,X\ket=0$ one has that $J$ is horizontal with respect to $X$. The subbundle $T_H^*M$ possesses the connection
  \bea
  	D_Y\xi=\nabla_Y\xi+gX\cdotp\bra\nabla_YX,\xi\ket~,~~~\xi\in T^*_HM~,~~~Y\in TM~,\nn
  \eea
  where $\nabla$ is the Levi-Civita connection.
  Hence the covariant derivative $D_Y$ is written as
  \bea
  	D_Y\xi_{\mu}=Y^{\rho}\partial_{\rho}\xi_{\mu}-Y^{\gs}\Gc^{\rho}_{\gs\mu}\xi_{\rho}-X_{\mu}Y^{\gs}J_{\gs}^{~\rho}\xi_{\rho}\,.
  \nn\eea
  In particular setting $Y=X$
  \bea
  	D_X\xi_{\mu}=X^{\rho}\partial_{\rho}\xi_{\mu}+(-J^{\rho}_{~\mu}+\partial_{\mu}X^{\rho})\xi_{\rho}=L_X\xi_{\mu}+J_{\mu}^{~\rho}\xi_{\rho}~.\nn\eea
  Thus $\xi$ can be pushed down if the right hand side vanishes
  \bea
  	0=D_X\xi=L_X\xi+J\xi~.\label{cond_I}
  \eea
  This is not quite the usual condition $L_X\xi=0$, but rather depends on the details of $J$.
  However one can write a different connection for $T^*_HM$ as
  \bea
  	D^{(n)}_Y\xi=\nabla_Y\xi+gX\cdotp\bra\nabla_YX,\xi\ket-n\bra Y,X\ket J\xi~,\label{LC_n}
  \eea
  which is valid since $\iota_X(J\xi)=0$ from the horizontality of $J$.
  Choosing $n=1$, we get the condition
  \bea
  	0=D^{(1)}_X\xi=L_X\xi\nn
  \eea
  for pushing down $\xi$.

As above $S^5$ can be used as an example. Then $J$ is a complex structure transverse to the Hopf fibre and so $J^2=-1$ on $T^*_HS^5$.
  This shows that for any integer $n$, the connection $D^{(n)}$ has holonomy $e^{2\pi i(n-1)}$ along the Hopf fibre, so that it is a valid choice of connection for pushing down $T_H^*S^5$ to $\BB{P}^2$.
  For $n=1$, the push down bundle is $T^*\BB{P}^2$, while for general $n$, it is the twisted $T^*\BB{P}^2\otimes{\cal O}(n-1)$.

\subsection{Reduction of the spin bundle}\label{sec_red_spin}

In this subsection we will consider the particular case of the spin bundle.  According to the general discussion above, we need a spin connection with trivial holonomy and we will push down sections $s$ satisfying $D_X s=0$.
  Note that the push down bundle may be a spin bundle twisted by some line bundle or even a $\spinc$ bundle.

  To write down a spin connection,  choose a vielbein $\{e^a\,|\,e^a\in\Gc(TM),~ \bra e^a,e^b\ket=\gd^{ab}\}$ and consider the Levi-Civita connecion in this basis
  \bea
  	\go_Y^{ab}=\bra e^a,\nabla_Y e^b\ket~,~~~Y\in\textrm{vect}\,(M)~.\nn
  \eea
  Then the spin connection is the lift $\FR{so}\to\FR{spin}$
  \bea
  	D_Y=Y\cdotp\partial+\frac14\go_Y^{ab}\Gc^{ab}~.\nn
  \eea

  As the spin bundle is equipped with a spinor Lie-derivative $L_Y^s$ for $Y$ Killing  \cite{Kosmann1971,Figueroa:1999va}, a natural requirement for pushdown could be $L^s_Xs=0$. We will see that this condition can be made precise along the same lines as in the discussion about the cotangent bundle above.

We can pick $\{e^a\}$ to satisfy locally
  \bea
 	 L_Xe^a=0~,\label{choice_precarious}
  \eea
  where $X$ is along the $S^1$ fibre and normalised as always. We first show that when this is done, then $L_X^s=X\cdotp\partial$, i.e. an ordinary derivative.
  The spinor Lie derivative along a Killing vector field is defined as
  \bea
  	 L^s_X=D_X+\frac14(\nabla_mX_n)\Gc^{mn}=X\cdotp\partial+\frac14\go_X^{ab}\Gc^{ab}+\frac14(\nabla_{\mu}X_{\nu})\Gc^{\mu\nu},\nn
  \eea
  where $\Gc_{\mu}=\Gc_ae^a_{\mu}$.
  Since $L_X e^a=0$ one has $\nabla_X e^a_{\mu}=-e^{a\nu}\nabla_{\mu} X_{\nu}=(Je^a)_{\mu}$ so that $\go_X^{ab}=\bra e^a,Je^b\ket$. Thus
  \bea L^s_X=X\cdotp\partial+\frac14\bra e^a,Je^b\ket\Gc^{ab}-\frac14J_{\mu\nu}\Gc^{\mu\nu}=X\cdotp\partial\,.\nn\eea
On the other hand, similarly to what we did in \eqref{LC_n},
one can modify the spin connection   \bea D\to D^{(n)}=D-\frac{n}{4}gX \slashed{J}\,,\nn\eea
  so that $D^{(1)}_X$ will coincide with $L_X^s=\partial_{\ga}$ when \eqref{choice_precarious} holds.

In what follows we shall use $D^{(1)}$ for the connection and check its holonomy along the circle fibre. Note that $D^{(1)}_X=L_X^s=\partial_{\ga}$ is a local expression and does not imply that the holonomy is 1.
  Indeed we have ignored the following global issue.
  Locally one adjusts the trivialisation of $TM$ to satisfy \eqref{choice_precarious}, but the adjustments may not be liftable to $\FR{spin}$.  A trivialisation of the spin bundle that it is independent of the $S^1$-fibre might not exist. In particular, when the entire  fibre does not lie in one patch, there could be a nontrivial transition function when going around the circle.
  An instance where this obstruction occurs is $S^{4k+1}\rightarrow \mathbb P^{2k}$.  In such case, the reduction of the bundle cannot proceed straightforwardly, but one may instead push down the spin bundle into a $\spinc$ bundle.

\subsection{Reduction of the Killing spinor on Sasaki-Einstein manifolds}\label{sec_Rots}

In this subsection we further specialize to the case where the 5D manifold is Sasaki-Einstein (SE). On any such manifold one can find Killing spinors
\bea
\label{KSSE}
	D_m\xi^1=-\frac{i}{2}\Gc_m\xi^1~,~~~~D_m\xi^2=+\frac{i}{2}\Gc_m\xi^2~,
\eea
and we are interested in establishing under which conditions these Killing spinors can be pushed down to the base.
We refer the reader to the appendix of \cite{2016arXiv160802966Q} for a review of the Sasaki-Einstein geometry that we need (one may also consult \cite{2010arXiv1004.2461S} for a more comprehensive view).

Consider a SE manifold with metric $g_{\mu\nu}$. We will make use of the Reeb vector field $\reeb$ and the contact 1-form $\gk=g \reeb$ satisfying $\iota_{\sreeb}d\gk=L_{\sreeb}d\gk=0$. We also need the complex structure $J$, acting on the plane transverse to $\reeb$, which is related to $\reeb$ by $\nabla_Y\reeb=JY$.
For the Sasakian geometry,  $J$ induces a K\"ahler structure transverse to the Reeb, i.e. $J$ satisfies the integrability condition
\bea
&&\bra Z,(\nabla_XJ)Y\ket=-\gk(Z)\bra X,Y\ket+\bra Z,X\ket\gk(Y)\label{integrability}~,
\eea
where $\bra-,-\ket$ is the inner product using the metric. We will use the same letter $J$ for the complex structure as well as for the 2-form $gJ$.
Finally the Sasaki-Einstein condition further implies that
\bea
R_{mn}=4g_{mn}~.\label{Einstein}
\eea

The Killing spinor equations \eqref{KSSE} can be solved using the approach of \cite{FriedrichKath}. Consider the rank 1 subbundle $W_{\mu}$ of the spin bundle $W$ consisting of $\psi$ satisfying
\bea
	{\reeb}\psi=-\psi~,~~~\big(\mu JX-\frac{i}{2}(1+{\reeb})X\big)\psi=0~,~~~\forall X\in\Gc(TM)~,\label{sub_bundle}
\eea
where $\mu=\pm1$ and we have omitted $\Gc$ whenever  Clifford multiplication is obvious.
One then defines a connection for $W_{\mu}$
\bea
	\tilde D_X=D_X+\frac{i\mu}{2}X~.\nn
\eea
This is indeed a connection, i.e. it preserves $W_{\mu}$, and furthermore it is flat when restricted to $W_{\mu}$ (more details can be found in \cite{Qiu:2013pta}).
If the SE manifold $M$ is simply connected, there is a unique (up to a constant multiple) solution to
\bea
	D_X\psi=-\frac{i\mu}{2}X\,\psi~,~~~~\mu=\pm1\,.\nn
\eea
Apart from \eqref{sub_bundle}, the solution satisfies
\bea
	 \slashed{J}\psi=-4i\mu\psi~.\label{spin_form_degree}
\eea

Since a section of the spin bundle can be reduced if $L_X^ss=D_X^{(1)}s=0$, we now turn to compute the Lie derivative of a Killing spinor.
The spinor Lie derivative along a Killing vector can be shown to satisfy the important properties
\bea
\label{killingmu}
	[L_X^s,Y\cdotp\Gc]=[X,Y]\cdotp\Gc~,~~[L_X^s,L_Y^s]=L^s_{[X,Y]}~,~~[L_X^s,D_Y]=D_{[X,Y]}~.
\eea
Using these one sees that the Lie derivative of a Killing spinor $\psi$ along a Killing vector $X$ is also Killing.  Using \eqref{sub_bundle}, \eqref{spin_form_degree}  one can show
\bea
\label{spinliepsi}
 	L_X^s\psi=\big(\frac{i\mu}{2}\bra X,\reeb\ket-\frac{i\mu}{8}\bra dX,J\ket-\frac14(\gk\wedge L_X\reeb)\cdotp\Gc\big)\psi\,.
 \eea
In the formulae above we routinely identify vectors with their dual 1-form and vice versa.

For the next subsection we can assume that the Killing vector $X$ commutes with $\reeb$
\bea
	L_X\reeb=0\,.\nn
\eea
In this case $L_X^s$ preserves the rank 1 subbundle \eqref{sub_bundle} so that $L_X^s\psi=i\mu f_X\psi$ for some constant $f_X$ (the details are in appendix B of \cite{Qiu:2013pta} \footnote{if one tries to check the calculation there, pay attention to the typo: the displayed equation before (87), $L_X^s=D_X-1/8\nabla_{[m}X_{n]}\Gc^{mn}$ should be $L_X^s=D_X+1/8\nabla_{[m}X_{n]}\Gc^{mn}$.}). The last term in \eqref{spinliepsi} is zero and hence
\bea
	f_X=\frac{1}{2}\bra X,\reeb\ket-\frac{1}{8}\bra dX,J\ket\,.\label{f_X}
\eea
Knowing that $f_X$ is a constant, this formula can evaluated at a convenient point.

As we stressed above we also need to compute the holonomy of $L_X^s$. This is best done without resorting to local computation. To this end we will introduce a spinor representation using horizontal forms (see also section 2.6 of \cite{2008CMaPh.280..611M}).

Using the Reeb vector $\reeb$, one can define the horizontal forms
\bea
	\go\in\Go^{\sbullet}_H(M)~~\textrm{iff}~~\iota_{\sreeb}\go=0~.\nn
\eea
Using the transverse complex structure $J$ one further decomposes $\Go^p_H=\oplus_{i+j=p}\Go^{i,j}_H$.
Now one can define the so called canonical $\spinc$-structure. Consider
\bea
W_{can}=\small{\textrm{$\bigoplus$}}\,\Omega_H^{0,\sbullet}(M)~.\label{can_spin_bundle}
\eea
One has a representation of the Clifford algebra on $W_{can}$: let $\psi$ be any section of $W_{can}$ and $\chi$ a 1-form, define the Clifford action
\bea
\chi\cdotp\psi=\Bigg\{
                             \begin{array}{cc}
                               \sqrt2 \chi\wedge\psi & \chi\in\Omega_H^{0,1}(M) \\
                               \sqrt2 \iota_{g^{-1}\chi}\psi & \chi\in\Omega_H^{1,0}(M) \\
                               (-1)^{\deg+1}\psi & \chi=\kappa \\
                             \end{array}.\label{can_spin_rep}
\eea
This in fact defines a priori a $\spinc$-structure whose characteristic line bundle (see chapter 5 in \cite{Salamon}) is the anti-canonical line bundle associated with the complex structure $J$.
This latter line bundle is trivial on $M$ for simply connected SE manifolds. Hence its square root is also a (trivial) line bundle, so that the $\spinc$ is in fact spin\footnote{In general, SE manifolds with $H_1(M,\BB{Z})_{tor}=0$, are spin (see theorem 7.5.27 in \cite{BoyerGalicki}).}.
With this concrete representation, the first condition in \eqref{sub_bundle} says that $\psi$ is in $\Go_H^{0,2k}$ while the second tells whether its (0,0) or (0,2) depending on $\mu$ (as also does \eqref{spin_form_degree}).

We mentioned above the characteristic line bundle of a $\spinc$-structure, which in our case is generated by $\Go_H^{0,2}$. For SE geometry this line bundle is trivialised by a nowhere vanishing section $\bar\varrho$ of $\Go^{0,2}_H$. Thanks to the triviality, one can identify $W\simeq \oplus\Go_H^{0,\sbullet}$. However one needs to remember that this is a statement at the level of topology, while for covariant derivatives, spinor Lie derivatives etc., the isomorphism $W\simeq \oplus\Go_H^{0,\sbullet}$ has a non-trivial effect. This is especially important for reducing the spin bundle, which we turn to next.

Pick a Killing spinor $\psi$ satisfying \eqref{killingmu} with $\mu=1$. Using this spinor one can write all other spinors by Clifford multiplying $\psi$ with $\Go_H^{0,\sbullet}$
\bea
	\xi=\eta\wedge\psi\in W~,~~~\eta\in\Go^{0,\sbullet}_H~.\nn
\eea
Let now $X=\partial_{\ga}$ be the vector field of the $U(1)$-fibration.
As we proved in section \ref{sec_red_spin}, if the vielbein on $M$ is invariant under $X$ then $D^{(1)}_X=L_X^s$, and so
\bea
	D^{(1)}_X(\eta\wedge \psi)=L^s_X(\eta\wedge \psi)=(L_X\eta)\wedge\psi+\eta\wedge L_X^s\psi~.\nn
\eea
As $X$ is induced from a circle action on $M$, the $L_X\eta$ term has the right period, so whether or not $D^{(1)}_X$ has trivial holonomy hangs on the last term $L_X^s\psi$.  For our purposes $L_X^s\psi=if_X\psi$ for a constant $f_X$. Thus $f_X\in\BB{Z}$ ensures that the holonomy is trivial.  When this  condition fails, the reduction is not impossible, but rather one might need to adjust the spin connection.

\subsection{Specialising to toric Sasaki-Einstein}\label{sec_SttSE}
In the toric setting $M$ has isometry $U(1)^3$ generated by $e_a,\,a=1,2,3$, and the Reeb vector is a constant combination of the three $U(1)$'s: $\reeb=\sum_{a=1}^3\reeb^ae_a$.
We also seek another combination $X=\sum_{a=1}^3X^ae_a,~X^a\in\BB{Z}$, so that $X$ has closed orbits of period $2\pi$ and $M$ is a regular foliation by the orbits. In other words $M$ is a $U(1)$-fibration over a 4D base $B$.

Let us investigate what requirement do we have on $X^a$ so that $f_X$ in \eqref{f_X} vanishes, that is $L_X^s\psi=0$. Note that in the current setting $L_X\reeb=0$ trivially.
Denoting with $\vec X=(X^1,X^2,X^3)$ the 3-vector parametrizing $X=X^ae_a$, we decompose (non-uniquely)
\bea
	\vec X=\sum_{i=1}^{\tt n}\gl_i\vec v_i~,\nn
\eea
where $\tt n$ is the total number of faces of the moment map cone of $M$. In fact it is possible to choose $\gl_i\in\BB{Z}$ since
\bea
	\pi_1(M)=0~~ \Leftrightarrow~~\textrm{span}\,\bra \vec v_1,\cdots\vec v_{\tt n}\ket=\BB{Z}^3~.\nn
\eea
Each $\vec v_i$ represents a $U(1)$ that vanishes of degree 1 at face $i$, and so by a local computation
\bea
	\bra dv_i,J\ket=-2~ , \nn
\eea
where we also use $v_i$ to denote the vector field $\sum_av_i^ae_a$.
This shows $f_{v_i}=1/2$ and $f_X=(1/2)\sum\gl_i$.
To formulate this quantity geometrically, we note that the SE condition implies that there exists a $\vec\xi\in\BB{Z}^3$ such that $\vec\xi\cdotp \vec v_i=1,~\forall i$. Then
\bea
	f_X=\frac12\sum_{i=1}^{\tt n}\gl_i=\frac12 \vec X\cdotp \vec\xi~.\nn
\eea
Hence  the spin bundle is reducible to $B$ if $\vec X\cdotp\xi=2\BB{Z}$.
Note that since such $\vec \xi$ must be primitive (its components have gcd 1), one may assume that $\vec\xi=[1,0,0]$.

The geometrical meaning of this condition is this: as the metric cone $C(M)$ over $M$ is a Calabi-Yau, it has a holomorphic volume form $\Go$.
From this one can construct a nowhere vanishing section $\varrho=\iota_{\sreeb}\Go\in\Go_H^{0,2}$.
Then $f_X$ is the charge of $\varrho$ under $X$.
For a geometry with $f_X=0$ we can then simply declare that the spin bundle on $M$ can be reduced to that of $B$.

\section{A classification result}\label{sec_Acr}
We first set up some nomenclature.
The geometry of $M$ is entirely encoded by a moment map cone $C_{\mu}(M)\subset\BB{R}^3$.
Let $\vec v_i\in\BB{Z}^3,~i=1,\cdots\tt m$ be the (primitive) inward pointing normals of the $\tt n$ faces of $C_{\mu}$.
Let $\reeb=\sum_{a=1}^3\reeb^ae_a$ be the Reeb vector field, and assume that $\vec\reeb$ is within the dual cone $C_{\mu}^{\vee}$, i.e.
\bea
	\vec\reeb=\sum_{i=1}^{\tt m}\gl_i\vec v_i~,~~~\gl_i>0~.\label{dual_cone}
\eea
With this assumption, the plane (where $y^a$ are the coordinate of $\BB{R}^3$)
\bea
	\big\{\vec y\in\BB{R}^3|\vec \reeb\cdotp \vec y=\frac12\big\} \nn
\eea
intersects $C_{\mu}$ at a convex polygon $\Gd_{\mu}$ if $C_{\mu}$ is convex.
Then the geometry of $M$ is that of a $U(1)^3$ fibration over $\Gd_{\mu}$, except that at each faces of $\Gd_{\mu}$, a certain $U(1)$ becomes degenerate.
More concretely if the normal associated with face $i$ is $\vec v_i$, then the $U(1)$ given by $\sum_{a=1}^3v^a_ie_a$ degenerates. \begin{figure}[h]
\begin{center}
\begin{tikzpicture}[scale=.5]
\draw [-,blue] (-1,-1) -- node[below] {\small$1$} (1,-1) -- node[right] {\small$2$} (1.5,.5) -- node[right] {\small$3$} (-.5,2) -- node[left] {\small$4$} (-1.5,1) -- node[left] {\small$5$} (-1,-1);
\draw (1,-1.6) ellipse (.3 and .6);
\draw (-1,-1.4) ellipse (.2 and .4);
\draw (1.5,0.8) ellipse (.18 and .3);
\draw (-0.5,2.3) ellipse (.1 and .3);
\draw (-2,1) ellipse (.5 and .2);
\end{tikzpicture}
\caption{The polygon $\Gd_{\mu}$. The circles represent the closed Reeb orbits. }\label{fig_toric_polygon}
\end{center}
\end{figure}
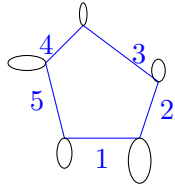
In particular, at the intersection of faces, only one $U(1)$ remains non-degenerate and its orbit is a closed Reeb orbit.
These are the only loci for closed Reeb orbits if $\vec\reeb$ is chosen generically.

We assume the following for $C_{\mu}$ (see \cite{2001math......7201L})
\begin{enumerate}
  \item Convexity (where $\vec v_{{\tt n}+1}:=\vec v_1$):
  \bea [\vec v_i,\vec v_{i+1},\vec v_k]=(\vec v_i\times \vec v_{i+1})\cdotp\vec v_k>0~,~~\forall k\neq i,i+1~.\label{convex}\eea
  \item Goodness\footnote{This condition was phrased in \cite{2001math......7201L} as: $\BB{Z}^3\cap\textrm{span}_{\BB{R}}\bra\vec v_i,\vec v_{i+1}\ket=\textrm{span}_{\BB{Z}}\bra\vec v_i,\vec v_{i+1}\ket$ for all $i$.}: $\exists \vec n_i\in\BB{Z}^3$, such that $[\vec n_i,\vec v_i,\vec v_{i+1}]=1$, $\forall i$.
  \item Gorenstein: $\exists \vec\xi\in\BB{Z}^3$ such that $\vec\xi\cdotp\vec v_i=1,~\forall i$,  see \cite{Martelli:2005tp}.
\end{enumerate}
The first condition is for compactness of $M$, the second for smoothness while the third guarantees the existence of a holomorphic volume form $\Go\in\Go^{3,0}(C(M))$, where $C(M)$ is the metric cone over of $M$.
In other words, the Gorenstein condition is the Calabi-Yau condition for the cone over $M$.

One may assume without loss of generality that $\vec\xi=[1,0,0]$, and so we write
\bea \vec v_i=\left[
                \begin{array}{c}
                  1 \\
                  x_i \\
                  y_i \\
                \end{array}
              \right].\label{used_I}\eea
Next let $\sum_{a=1}^3X^ae_a$ represent the vector field $X$, we then have the correspondence
\begin{proposition}
  The 5D toric Sasaki-Einstein manifolds with a freely acting $U(1)$ that preserves $\Go$ are in 1-1 correspondence (up to $SL(3,\BB{Z})$ transformation) with convex 2D-polygons whose vertices $(x_i, y_i)$ are in $\BB{Z}^2$, and furthermore the $x$-coordinate of neighbouring vertices must differ by $\pm1$. This implies that the number of vertices is even ${\tt m}= 2{\tt n}$. If one requires $\pi_1=0$, then all the $y_i$'s should have greatest common divisor 1.
\end{proposition}
\begin{proof}
That $X$ acts freely means that at the intersection of face $i,i+1$, one has
\bea
	\det[\vec X,\vec v_i,\vec v_{i+1}]=\vec X\cdotp(\vec v_i\times\vec v_{i+1})=\pm1~,\label{freeness}
\eea
so that not only the vector field $X$ is nowhere zero, but its stability group is trivial for all points.
This also ensures the smoothness of $M$ since \eqref{freeness} implies goodness.

We focus on the case $f_X=\vec\xi\cdotp \vec X=0$, then with a further $SL(3,\BB{Z})$ transformation one can assume
\bea
 \vec X=\left[
                \begin{array}{c}
                  0 \\
                  0 \\
                  1 \\
                \end{array}\right]~,\nn
 \eea
while preserving all the other assumptions we have made so far\footnote{Keep in mind that if $\vec v_i$ is transformed with $g\in SL(3,\BB{Z})$, then $\vec\xi$ is transformed with $g^T$.}.
With these assumptions \eqref{freeness} says
\bea
	x_i-x_{i+1}=\pm1~,\label{freeness_I}
\eea
and the convexity \eqref{convex} says
\bea
\det\left[
  \begin{array}{cc}
    x_i-x_k & x_{i+1}-x_i \\
    y_i-y_k & y_{i+1}-y_i \\
  \end{array}\right]>0~.\label{convex_I}
\eea
It is not difficult to see that the solution to \eqref{freeness_I}, \eqref{convex_I} are labelled by a convex polygon on the $x-y$ plane, for which the $x$-coordinates of successive vertices differ by $\pm1$.

Finally for the toric manifolds considered $\pi_1=\BB{Z}^3/\textrm{span}\,\bra\vec v_1,\cdots, \vec v_{2\tt n}\ket$, so from the explicit form of the $\vec v_i$'s, this is realised if $\gcd(y_i)=1$.
\end{proof}

To fix the $SL(3,\BB{Z})$ redundancy, we enforce
  \begin{enumerate}
    \item the entire polygon lies to the right of $y$ axis
    \item vertex 1 and 2 are on $(0,0)$ and $(1,0)$
    \item $\det[\vec e_1,\vec{e}_{{\tt n}+1}]\geq0$, and if $\det[\vec e_1,\vec{e}_{{\tt n}+1}]=0$, then $\det[\vec e_2,\vec{e}_{{\tt n}+2}]\geq0$ and so on,
  \end{enumerate}
where $\vec e_i$ denotes the edge from vertex $i$ to vertex $i+1$.

Indeed using a cyclic permutation, one fixes the vertex with the smallest $x$-value as the $1^{st}$ vertex, satisfying item 1 one the list above.
The $SL(3,\BB{Z})$ redundancy now consists of lower triangular matrices only. With these, one can set $(x_1,y_1)=(0,0)$, and a further transformation sets $(x_2,y_2)=(1,0)$, satisfying item 2 of the list. If the resulting polygon does not satisfy item 3 we can act as follows. First flip the sign of all $x_i,y_i$ ($X$ is now $[0;0;-1]$, but this does not affect anything). We can now repeat the steps above and make the polygon satisfy item 1, 2 and 3. This corresponds essentially to turning the polygon around so that the $({\tt n}+1)^{th}$ vertex (the right most one) becomes the first one.  The first two pictures of figure \ref{fig_hexagon} provide an explicit example of this flip.

\begin{example} [$Y^{p,q}$-spaces] \label{example:Ypq}
  Take a quadrangle with vertices placed at $[0,0],[1,0],[2,p-q],[1,p]$, with $p>q> 0$ and $\gcd(p,q)=1$,  i.e. the normals are
  \bea  \label{eq:Ypqnormals}
  [\vec v_1,\cdots,\vec v_4]=\left[
       \begin{array}{cccc}
         1 & 1 & 1 & 1 \\
         0 & 1 & 2 & 1 \\
         0 & 0 & p-q & p \\
       \end{array}\right]\,.\label{ex_Ypq}\eea
Note that the metric cone in this case can be obtained by a K\"ahler reduction of $\BB{C}^4$ with a $U(1)$ of weight $[-p,p+q,-p,p-q]$, c.f. section 4 of \cite{Boyer:2011ia}.
From the explicit metric for $Y^{p,q}$  \cite{Gauntlett:2004yd}, that we write down in appendix \ref{app:Ypqexample}, the $U(1)$ fibration is obvious.
In contrast, $L^{a,b,c}$ spaces \cite{Cvetic:2005ft} do not offer any free $U(1)$ and, if one writes down the normals, one sees that they do not fall into our classification.

Figure \ref{fig_hexagon} shows the normals of $Y^{2,1}$, as well as a hexagon example. For the hexagon, from the vertices we read off the normals
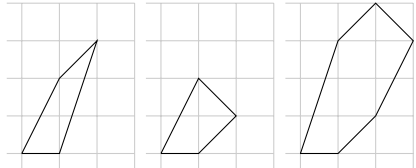
\begin{figure}[h]
\begin{center}
\begin{tikzpicture}[scale=0.5]
\draw [step=1,thin,gray!40] (-.4,-.4) grid (3,4);
\draw[-] (0,0) -- (1,0) -- (2,3) -- (1,2) -- (0,0);
\end{tikzpicture}
\begin{tikzpicture}[scale=0.5]
\draw [step=1,thin,gray!40] (-.4,-.4) grid (3,4);
\draw[-] (0,0) -- (1,0) -- (2,1) -- (1,2) -- (0,0);
\end{tikzpicture}
\begin{tikzpicture}[scale=0.5]
\draw [step=1,thin,gray!40] (-.4,-.4) grid (3,4);
\draw[-] (0,0) -- (1,0) -- (2,1) -- (3,3) -- (2,4) -- (1,3) -- (0,0);
\end{tikzpicture}\caption{The first two are equivalent polygons representing the $Y^{2,1}$ space, and the last is a hexagon example}\label{fig_hexagon}
\end{center}
\end{figure}
\bea
 [\vec v_1,\cdots,\vec v_6]=\left[
         \begin{array}{cccccc}
           1 & 1 & 1 & 1 & 1 & 1\\
           0 & 1 & 2 & 3 & 2 & 1\\
           0 & 0 & 1 & 3 & 4 & 3\\
         \end{array}\right]\label{ex_hex}.
   \eea

Here is an octagon example
\bea
[\vec v_1,\cdots,\vec v_8]=\left[
         \begin{array}{cccccccc}
           1 & 1 & 1 & 1 & 1 & 1 & 1 & 1\\
           0 & 1 & 2 & 3 & 4 & 3 & 2 & 1\\
           0 & 0 & 1 & 3 & 6 & 6 & 5 & 3\\
         \end{array}\right]~.\nn
\eea
Note that the polygons appearing here should not be confused with the polygons $\Gd_{\mu}$.
\end{example}

\subsection{The geometry of the base}\label{sec_tgotb}
We fix the orientation of the 5-manifold by picking the volume form
\bea
	\textrm{Vol}_M=\frac18 \gk \wedge d\gk\wedge d\gk~.\nn
\eea
Since the vector field $X$ is  everywhere nonzero we fix the volume form of $B$ as
\bea
	\textrm{Vol}_B=\iota_X\textrm{Vol}_M~.\label{vol_B}
\eea
At the intersection of two faces, there will be only one nondegenerate $U(1)$.
Thus, $\reeb$ and $X$ both being linear combinations of $U(1)$'s, must (anti)align at these loci.
At the intersection of face $i$ and $i+1$, the three weights $\vec\reeb,\vec v_i,\vec v_{i+1}$ always form a right-handed base. Indeed from the condition \eqref{dual_cone} one has
\bea
	[\vec \reeb,\vec v_i,\vec v_{i+1}]=\sum_{i=1}^{\tt n}\gl_j[\vec v_j,\vec v_i,\vec v_{i+1}]>0~.\nn
\eea
The right hand side is greater than zero from \eqref{convex}.
On the other hand $[\vec X,\vec v_i,\vec v_{i+1}]=\pm1$, thus we conclude
\bea
&&[\vec X,\vec v_i,\vec v_{i+1}]=+1~,~~~\reeb~\textrm{and}~X~\textrm{parallel},\nn\\
&&[\vec X,\vec v_i,\vec v_{i+1}]=-1~,~~~\reeb~\textrm{and}~X~\textrm{anti-parallel,}\nn
\eea
at the locus corresponding to the intersection of face $i$ and $i+1$.
Note that in the polygon picture of the normals, the $+1$ occurs for the sides of the polygon where the x-coordinate increase, and the $-1$ when it decreases (going around the polygon counter-clockwise).
So they will occur the same number of times, which also is a way of seeing that $f_X = 0$.

Due to this misalignment of $X$ with respect to $\reeb$ across the manifold, the orientation of $B$ determined according to \eqref{vol_B} does not always agree with that of $d\gk\wedge d\gk/8$.
At a corner where $[\vec X,\vec v_i,\vec v_{i+1}]=-1$, the orientation of $B$ is opposite to that of the transverse plane field of $M$.
This will have important effect when we consider instantons, since the (anti-)self-duality condition depends on the choice of volume form, more about this in section \ref{sec_Datis}.

\subsection{Intersection form and geometry of $B$}\label{sec_IfagoB}
To understand the geometry of the base manifold, we compute the pairing of $H_2(B,\BB{Z})$.
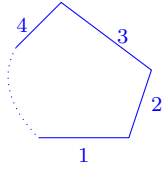
\begin{figure}[h]
\begin{center}
\begin{tikzpicture}[scale=.6]
\draw [-,blue] (-1,-1) -- node[below] {\scriptsize$1$} (1,-1) -- node[right] {\scriptsize$2$} (1.5,.5) -- node[right] {\scriptsize$3$} (-.5,2) -- node[left] {\scriptsize$4$} (-1.5,1);
\draw [dotted,blue] (-1.5,1) to [out=-120, in =135] (-1,-1);
\end{tikzpicture}
\caption{The momentum polytope of a 5D toric contact manifold.}\label{fig_geom_base}
\end{center}
\end{figure}
Figure \ref{fig_geom_base} represents the base of the moment map polytope of the 5D toric manifold. Taking a further quotient along $X$ gives the base $B$.
Note that in the classification above $B$ is \emph{not} toric K\"ahler, we are merely using the polytope for $M$ to visualise the geometry of $B$.

The edges in figure \ref{fig_geom_base} generate $H^2(M,\BB{Z})$; in fact each edge corresponds to a torus invariant 3D submanifold (some lens space).
Taking the quotient along $X$, we get a generating set for $H^2(B,\BB{Z})$.
There are relations among the generators. Denoting by $[x_i]\in H^2(B,\BB{Z})$ the generator associated with edge $i$, we have
\bea
\vec r\cdotp\sum_{i=1}^{2\tt n}\vec v_i[x_i]=0~,~~~\forall\vec r\in\BB{Z}^3~,~{\rm such}~{\rm that}~~\vec r\cdotp\vec X=0~.\label{rel_div}
\eea

As we saw in the last section, one can assume that $\vec X$ has been set to be $[0;0;1]$ and the normals have been put in the standard form
\bea \label{eq:normalscanonical}
[\vec v_1,\cdots,\vec v_{\tt n},\vec v_{{\tt n}+1},\vec v_{2{\tt n}}]=\left[
\begin{array}{cccccccc}
    1& 1 & 1 & \cdots & 1 & 1 & 1 & 1\\
    0 & 1 & 2 & \cdots & {\tt n} & {\tt n}-1 & \cdots & 1\\
    0 & 0 & * & * & * & * & * & *\\
\end{array}\right],\label{std_form}\eea
then the relation \eqref{rel_div} is simply
\bea \sum_{i=1}^{2\tt n}v^a_i[x_i]=0~,~~~a=1,2~.\nn
\eea
Looking at the first and second row of \eqref{std_form}, we can take $[x_i],~i=3,\cdots,2{\tt n}$ as a free generating set of $H^2(B,\BB{Z})$.

The intersection form of $H^2(B,\BB{Z})$ can be computed as the intersection number of the $[x_i]$'s, which is
\bea
 \bra [x_i],[x_{i+1}]\ket&=&\sgn[\vec X,\vec v_i,\vec v_{i+1}]\nn\\
\bra [x_i],[x_i]\ket&=&-\sgn[\vec X,\vec v_{i-1},\vec v_i]\,\sgn[\vec X,\vec v_i,\vec v_{i+1}]\,[\vec X,\vec v_{i-1},\vec v_{i+1}]\label{intersection_form}\eea
and zero otherwise. Here the orientation we used for the pairing is that of \eqref{vol_B}.

\begin{example}
\label{YPQ}
Take the $Y^{p,q}$ spaces as an example, the normals are in \eqref{ex_Ypq}, and so the paring matrix between $[x_3,x_4]$ is
\bea \bra-,-\ket_{Y^{p,q}}=\left[
       \begin{array}{cc}
         0 & -1 \\
         -1 & 2 \\
       \end{array}\right]\nn\eea
This pairing matrix is equivalent to the standard form
\bea H=\left[
       \begin{array}{cc}
         0 & 1 \\
         1 & 0 \\
       \end{array}\right],\label{std_int_form}\eea
that is, there is a matrix $g\in SL(2,\BB{Z})$ such that $\bra g-,g-\ket=H$. Note that $H$ is the intersection form of $S^2\times S^2$.

Take now a hexagon example \eqref{ex_hex}
\bea \bra-,-\ket=\left[
  \begin{array}{ccccc}
    -2 & 1 \\
    1 & 0 & -1\\
      & -1 & 2 & -1\\
     &  & -1 & 2\\
             \end{array}\right]\nn\eea
which is equivalent to $H\oplus H$, i.e. the intersection form of $\#_2(S^2\times S^2)$.

By the notation $\#_k(S^2\times S^2)$, we mean the connected sum of $k$ copies of $S^2 \times S^2$.
The connected sum of two manifolds joins them together near a chosen point of each, i.e. we delete a ball inside each manifold and glue together the resulting boundary spheres. Although the construction depends on the choice of balls, the result is unique up to diffeomorphism.

\end{example}
\begin{proposition}
  All manifolds appearing in the classification above are homeomorphic to $\#_k(S^2\times S^2)$ with $k= {\tt n} +1$.
\end{proposition}
\begin{proof}
We claim that all our intersection matrices are equivalent to a direct sum of terms $H$ of \eqref{std_int_form}.
If this is so, then by a famous theorem of Freedman (theorem 1.5 \cite{freedman1982}), there is a unique simply connected 4-manifold whose intersection form realizes the given quadratic form.
This shows that the manifolds in question have to be $\#_k(S^2\times S^2)$.

Next we prove the claim.
We always assume that the normals are put in the standard form of equation \eqref{std_form}; which makes the intersection form take the general form
\bea
\left[
  \begin{array}{ccccc}
    \ddots & 1 \\
    1 & -2 & 1 \\
      &  1 & 0 & -1 \\
      &    & -1  & 2 & -1\\
      &    &   & -1 & \ddots \\
        \end{array}\right].\nn
\eea
It is easy to see that this paring has even parity, i.e. $\bra x,x\ket=\textrm{even}$ for any $x$.

Working over $\BB{Q}$ and using elementary row and column operations, one can show that the pairing is equivalent to the diagonal matrix
\bea
 \textrm{diag}[-2,-\frac{3}{2},-\frac{4}{3},\cdots,-\frac{{\tt n}-1}{{\tt n}-2},\frac{{\tt n}-2}{{\tt n}-1},\frac{{\tt n}-3}{{\tt n}-2},\cdots,\frac12,-\frac12,2],\nn\eea
from which we see that its signature (the number of positive eigenvalues minus the number of negative eigenvalues) is zero.
Moreover the determinant of the intersection form is $(-1)^{{\tt n}-1}$ (since the matrices of the elementary row/column operations have determinant 1, one can compute the determinant using the above diagonal form) and hence it is invertible.
Thus, our intersection form is of maximum rank, is even and of zero signature.
It is easy to see that the same holds for the direct sum of factors of $H$, so by a theorem classifying the indefinite even quadratic forms (theorem 5.3 in chapter 2 of \cite{Milnor73}), they are equivalent.
\end{proof}
\begin{remark}
   Note that the complex structure of the resulting 4-manifold is \emph{not} inherited from the transverse complex structure of the 5-manifold, in contrast to the ones appearing below.
\end{remark}
\begin{remark}
 The manifolds $\#_k(S^2\times S^2)$ are a sub-family of
   \bea (\pm M_{E_8})^{\#2m}\#(S^2\times S^2)^{\#k}\label{general_spin}\eea
   where $M_{E_8}$ is some 4-manifold with intersection form the Cartan matrix of $E_8$. One has that any simply connected smooth 4-manifold has the \emph{homeomorphism} type above (however the converse statement is an open problem). Indeed, the intersection form of a spin 4-fold must be indefinite, for by Donaldson's theorem, a definite intersection form can be diagonalised to $+1$ or $-1$ and so not spin (since the intersection form of spin manifolds have even parity). Then the classification of the indefinite forms gives $\pm nE_8\oplus kH$. Furthermore the number of copies of $E_8$ is even so that the intersection form has signature divisible by 16 according to Rohklin's theorem. And if $m>0$, one needs $k>0$ so as not to have a definite form, leading to \eqref{general_spin}.
\end{remark}

\subsection{More examples not included in the classification}\label{Meniitc}

If one gives up the condition $L_X \Go =f_X = 0$ or equivalently $\vec \xi \cdotp \vec X = 0$, one can find some more sporadic cases.
We do not consider these in this paper, leaving them for future study, but we make the following observation.
Consider the condition
\[
 \vec X \cdotp (\vec v_i \times \vec v_{i+1} ) = \pm 1~.
\]
The vector $\vec w_i \equiv \vec v_i \times \vec v_{i+1}$ is a generator of our cone; and thus it is also a normal vector of the dual cone.
This means that we can think of the condition $\vec X \cdotp \vec w_i = \pm 1 \ \forall i$ as a ``generalized Gorenstein condition'' for the dual cone.
If we require to have strictly $\vec X \cdotp \vec w_i = +1$, it is exactly the Gorenstein condition for the dual cone.
Cones with this property, i.e. where both the cone and its dual are Gorenstein, are called \emph{reflexive Gorenstein}. They are well studied \cite{Skarke:2012zg, CoxKatz} , since they are important and useful in the context of mirror symmetry: the cone and its dual give us a mirror pair of CY manifolds.

Reflexive Gorenstein cones are in one-to-one correspondence with reflexive polytopes (polytopes that contain exactly 1 interior lattice point). In 2D there are 16 such polytopes (up to $GL_2(\mathbb{Z})$ transformations).
In figure \ref{fig_reg_SE} we have plotted some of these polytopes.
In contrast to the previous examples, the base $B$ is now a toric K\"ahler manifold.
We also note that now $X$ is always aligned with $\reeb$ at the loci of the closed Reeb orbits.

Depending on the details of the geometry, one may be able to push the spin bundle from $M$ to a spin or $\spinc$ bundle on $B$.
\begin{figure}[h]
\begin{center}
\begin{tikzpicture}[scale=.6]
\draw [step=1,thin,gray!40] (-1.2,-1.2) grid (2.2,2.2);
\draw[-,blue] (-1,-1) -- (2,-1) -- (-1,2) -- (-1,-1);
\node at (-1.3,.6) {\scriptsize1};
\node at (.6,-1.3) {\scriptsize2};
\node at (.65,.65) {\scriptsize3};
\end{tikzpicture}
\begin{tikzpicture}[scale=.6]
\draw [step=1,thin,gray!40] (-.2,-.2) grid (2.2,2.2);
\draw[-,blue] (0,0) -- (2,0) -- (2,2) -- (0,2) -- (0,0);
\node at (-.25,1) {\scriptsize1};
\node at (1,-0.3) {\scriptsize2};
\end{tikzpicture}
\begin{tikzpicture}[scale=.6]
\draw [step=1,thin,gray!40] (-.2,-.2) grid (2.2,3.2);
\draw[-,blue] (0,0) -- (2,0) -- (2,1) -- (0,3) -- (0,0);
\draw[->] (0,1) node[left] {\scriptsize$\vec v_1$} -- (.5,1);
\draw[->] (1,0) node[below] {\scriptsize$\vec v_2$} -- (1,0.5);
\draw[->] (2,0.6) node[right] {\scriptsize$\vec v_3$} -- (1.5,0.6);
\draw[->] (1,2) node[right] {\scriptsize$\vec v_4$} -- (.7,1.7);
\end{tikzpicture}
\begin{tikzpicture}[scale=.8]
\draw [step=1,thin,gray!40] (-.2,-.2) grid (2.2,2.2);
\draw[-,blue] (0,1) -- (0,0) -- (2,0) -- (2,2) -- (1,2) -- (0,1);
\node at (-.2,0.5) {\scriptsize1};
\node at (1,-0.3) {\scriptsize2};
\end{tikzpicture}
\begin{tikzpicture}[scale=.8]
\draw [step=1,thin,gray!40] (-.2,-1.2) grid (2.2,1.2);
\draw[-,blue] (0,0) -- (0,-1) -- (1,-1) -- (2,0) -- (2,1) -- (1,1) -- (0,0);
\node at (-.2,-0.5) {\scriptsize1};
\node at (.5,-1.3) {\scriptsize2};
\end{tikzpicture}
\caption{More examples that correspond to regular toric SE manifolds.}\label{fig_reg_SE}
\end{center}
\end{figure}
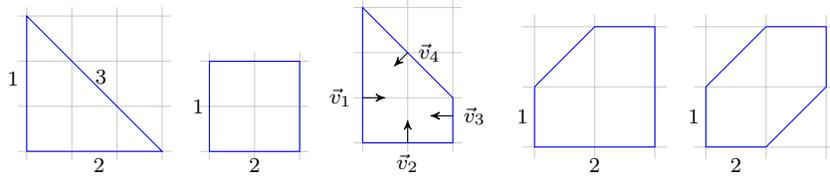

To give a bit more detail about these cases, let us work out some details of the examples in figure \ref{fig_reg_SE}.
The second and third one have normals given by
\bea [\vec v_1,\cdots,\vec v_4]=\left[
         \begin{array}{cccc}
           1 & 1 & 1 & 1 \\
           1 & 0 & -1 & k \\
           0 & 1 & 0 & -1 \\
         \end{array}\right],~~k=-1,0~ . \nn
 \eea
Denoting by $[i]$ the divisor of the $i^{th}$ face of the moment polygons in the figure above, the canonical classes (and also the K\"ahler class) are
\bea
&&3[3]~,\nn\\
&&2[3]+2[4]~,\nn\\
&&2[3]+3[4]~,\nn\\
&&2[3]+2[4]+[5]~,\nn\\
&&[3]+2[4]+2[5]+[6]~,\nn
\eea
respectively.
Here we have used relations among the divisors to eliminate $[1],\,[2]$.
Only in the second case is the canonical class divisible by 2 and one can reduce to a spin structure (the geometry is $S^2\times S^2$ after all).
For the rest, one gets $\spinc$ structures.

\section{Reduction of  $\mathcal{N}=1$ SYM to 4D}\label{sec_Red}

\subsection{Reduction of the action}\label{sec_Rota}

Starting from the works \cite{Kallen:2012cs,  Hosomichi:2012ek, Kallen:2012va} the
5D $\mathcal{N}=1$ supersymmetric Yang-Mills on a Sasaki-Einstein manifold has been constructed in \cite{Qiu:2013pta, Qiu:2014oqa}.
  For $\mathcal{N}=1$ vector multiplet the action has the following form  (we refer the reader to appendix \ref{app_spinors} where the conventions used here are spelled out)
\be \label{action_vector} \begin{split}
 S_{vec}= \frac{1}{(\gYMf)^2} & \int\limits_M\textrm{Vol}_M\,\Tr\Big[\frac{1}{2}F_{mn} F^{mn} -D_m\gs  D^m\gs-\frac12 D_{IJ}D^{IJ}+2 \gs t^{IJ}D_{IJ}- 10
 t^{IJ}t_{IJ}\gs^2 \\
&+i\gl_I\Gc^mD_m\gl^I-\gl_I[\gs,\gl^I]-i\, t^{IJ}\gl_I\gl_J\Big]~,
\end{split}
\ee
where the $\gl^I$ satisfy the symplectic Majorana condition
\bea
(\gl^I)^*=\ep_{IJ}\,C\,\gl^J~\nn~,
\eea
with $C$ being the charge conjugation matrix. The supersymmetry transformations read
\bea
\label{susytr}
&&\gd A_m = i\xi_I\Gc_m\gl^I~,\nn\\
&&\gd\gs = i\xi_I\gl^I~,\nn\\
&& \gd\gl_I = -\frac12(\Gc^{mn}\xi_I)F_{mn}+(\Gc^m\xi_I)D_m\gs-\xi^JD_{JI}+2 t_I^{~J}\xi_J\gs~, \label{susy_vect} \\
&&\gd D_{IJ} = -i\xi_I\Gc^mD_m\gl_J+[\gs,\xi_I\gl_J]+i t_I^{~K}\xi_K\gl_J+(I\leftrightarrow J)~.\nn
\eea
The spinor $\xi_I$ is also symplectic Majorana and satisfies the Killing equation
\bea
\label{KSEse}
 \nabla_m\xi_I= t_I^{~J}\Gc_m\xi_J~,~~~t_I^{~J}=\frac{i}{2 r}(\sigma_3)_I^{~J}~,~~~(\xi_I\xi_J)=-\frac12\ep_{IJ}~,\label{killing_eqn_I}
 \eea
where $\sigma_3=\rm{diag}[1,-1]$. The supercharge squares to a translation along the Reeb vector ${\reeb}^n= \xi^I \Gamma^n \xi_I$.
 In formula \eqref{killing_eqn_I}  $r$ is a dimensionful parameter corresponding to the size of the manifold. Explicitly the metric on $M$ is taken taken to be $r^2 ds_5^2$. In the limit $r \rightarrow \infty$ the theory approaches ${\cal N}=1$ in flat space.

We assume now that $M$ is a $U(1)$ bundle over a 4-manifold $S^1\to M\stackrel{\pi}{\to} B$, satisfying the conditions described in the previous sections, and rewrite the action as a SYM theory over $B$.
We set up some notation first. Let the $U(1)$ action be generated by $\partial_\alpha$.  The metric $g_M$ on $M$ is invariant along the fiber, that is $\partial_\alpha$ is a Killing vector. The metric can then be written in the following form,
\be
\label{eq:met5D}
r^2 ds_5^2 = r^2 \left ( ds_4^2 + e^{2\phi} ( d\alpha + b )^2\right )~,
\ee
where $r^2 ds_4^2$ is the metric on the 4D base $B$, $\alpha$ is the coordinate along the fiber, $b$ is the connection one-form for the fibration and $r e^{\phi}$ is the radius of the fiber. Because $\partial_\alpha$ is Killing both $\phi$ and $b$ are constant along the flow it generates.
In appendix \ref{app:Ypqexample} following \cite{Gauntlett:2004yd}, we present the metric of the $Y^{p,q}$ spaces, considered in example \ref{YPQ}, in this form.

Generically $e^\phi$ is non-constant over the base, and the five dimensional volume form $\textrm{Vol}_M$ is related to the volume form $\textrm{Vol}_4$ on the base by (note that this is not the same as $\textrm{Vol}_B$ defined in \eqref{vol_B})
\be \label{eq:Vol4can}
\textrm{Vol}_M= r e^\phi d\alpha\wedge\textrm{Vol}_4~.
\ee
Finally let $\gb=r^{-1} e^{-2\phi} g_M \partial^{\ga}=r(d\alpha+b)$.
Since $\partial_\alpha$ is Killing, $d\gb$ is constant along its flow, that is $L_{\partial_{\ga}}d\gb=0$. Additionally $\iota_{\partial_\alpha}(d\beta)=0$ and hence we can regard $d\gb$ as a 2-form on $B$.

In reducing we will take all the fields in the theory to be invariant under the (spinor) Lie-derivative along $\ga$. As explained in the previous sections they can then be regarded as fields on the base $B$. In particular we restrict the gauge bundle on $M$ to be one pulled back from $B$, so that only  gauge connections of the form $\pi^*A+ {\tilde \varphi }\gb$ are considered. Here $ {\tilde \varphi }$ is an adjoint scalar that is constant along the fiber. Note that  $ {\tilde \varphi }\gb$ is a globally defined adjoint valued 1-form and hence does not affect the topology type of the bundle.

This restriction on the fields is compatible with the supersymmetry transformations \eqref{susytr} as long as the spinor parameters $\xi_I$ are constant along the fiber. Under this condition the reduction gives rise to a supersymmetric field theory on the base $B$.

The four dimensional supersymmetry variation parameter $\xi_I$ satisfies
\be
\label{KSEred}
\begin{split}
\nabla_{\mu} \xi_I&=-{1\over 4} e^{\phi}d\beta_{\mu\nu}\gamma^{\nu}\gamma_5\xi_I + {t_I}^J \gamma_{\mu} \xi_J\,,\cr
0&= {1\over 2} \partial_{\mu}\phi \gamma^{\mu} \xi_I-{1\over 8} e^\phi d\beta_{ \mu\nu } \gamma^{\mu\nu} \gamma_5 \xi_I- {t_I}^J \xi_J\,,
\end{split}
\ee
where we use $\gamma_{\mu}$ for the gamma matrices in four dimensions and we regard $d\gb$ as a form on $B$. The first equation above matches with the generalized Killing spinor equation stemming from the rigid limit of ${\cal N}=2$ Poincar\`e supergravity \cite{Butter:2015tra}. The second equation is a constraint arising from the higher dimensional Killing spinor equation  \eqref{KSEse} along the fiber direction.

Plugging $\pi^*A+  {\tilde \varphi } \gb$ into the curvature we obtain ($\bra\sbullet,\sbullet\ket_M$ is the contraction using the 5D metric $g_M$, while $\bra\sbullet,\sbullet\ket_B$ uses $g_B$)
\bea
&&F_5=F_4+(D {\tilde \varphi })\wedge\gb+ {\tilde \varphi } d\gb~,\nn\\
&&\bra F_5,F_5\ket_M=\bra F_4+ {\tilde \varphi } d\gb,F_4+ {\tilde \varphi } d\gb\ket_B+2e^{-2\phi} \bra D {\tilde \varphi },D {\tilde \varphi }\ket_B~,\nn\eea

Making use of \eqref{KSEred} and setting $\varphi= e^{-\phi} {\tilde \varphi}$ the reduced supersymmetry transformations are given by
\bea
\label{reducedsusy}
&&\gd A_{\mu} = i\xi_I\gamma_{\mu}\gl^I~,\nn\\
&&\gd \varphi = i\xi_I\gamma_5\gl^I~,~~~~\gd\gs = i\xi_I\gl^I~,\nn\\
&& \gd\gl_I = -\frac14\left(2\cancel{F} + \varphi e^\phi \cancel{d\beta} \right)\xi_I+(\cancel{D}\sigma+ \gamma_5 \cancel{D}\varphi)\xi_I - i [ \varphi, \sigma] \gamma_5 \xi_I \\
&& \qquad\quad -D_{IJ}\xi^J+2(\sigma +\gamma_5 \varphi) t_I^{~J}\xi_J~, \nn\\
&&\gd D_{IJ} = -i\xi_I\cancel{D}\gl_J-[\varphi,\xi_I\gamma_5\gl_J]+[\gs,\xi_I\gl_J]+(I\leftrightarrow J)\,. \nn
\eea
These are a specific instance of those arising from rigid ${\cal N}=2$ supergravity \cite{Hama:2012bg,Klare:2013dka,Pestun:2014mja,Butter:2015tra}. The 5D supersymmetry transformations \eqref{susytr} provide a compact packaging of the 4D ones.
Finally the 4D action reads (we suppressed subscripts ${}_4,\,{}_B,$)
\be \label{action_vector_4D}
 \begin{split}
& S^{4D}_{vec}= \int\limits_B\textrm{Vol}_4  \frac{ 2\pi r e^\phi}{(\gYMf)^2}\, \Tr\Big[\frac{1}{2}\bra F+ {\tilde \varphi } d\gb,F+ {\tilde \varphi } d\gb\ket+ \bra D { \varphi },D { \varphi }\ket-\varphi^2 \nabla^2 \phi-\bra D\gs,D\gs\ket
+ [ { \varphi },\gs]^2\\
&\hspace{0.5cm} -\frac12 D_{IJ}D^{IJ}
 +2 \gs t^{IJ}D_{IJ}- 10
 t^{IJ}t_{IJ}\gs^2+i\gl_I\slashed{D}\gl^I + \frac{i} {8} e^{\phi} \lambda_I\, \cancel{d\beta} \gamma_5 \lambda^I
 + \frac{ i } {2} \lambda_I\, \cancel{\partial\phi} \lambda^I \\
& \hspace{0.5cm} -\gl_I[\gs - \gamma_5\varphi,\gl^I]-i t^{IJ}\gl_I\gl_J\Big],
 \end{split}
 \ee

We see that after reduction the field theory defined by \eqref{action_vector_4D} has a position dependent YM coupling constant, the dependence coming from $e^{\phi}$. This is expected since we know that, when performing a Kaluza-Klein type reduction, the YM coupling picks up a factor of the radius of the $S^1$ fiber, which in our case is not of constant size. Nevertheless this theory is supersymmetric by construction.
We define the 4D YM coupling in terms of the 5D as
\be \label{eq:4DgYM}
	\frac{1}{{\gYMt}(x)} = \frac{2\pi r e^\phi}{(\gYMf)^2} \ .
\ee
 We want to point out that the action above reverts to the flat space SYM when $r\to \infty$. To see this one needs to remember that the geometric quantities such as the metric, $\gb$ and $t$ contain $r$ explicitly, while derivatives of the conformal factor goes to zero since $\phi$ is slow varying across distances far smaller than $r$.

With the goal of reaching a more conventional theory, in section \ref{sec_Defo} we will study diverse deformations of \eqref{action_vector_4D}. Along the way we will see that these deformations are $Q$-exact and hence do not affect supersymmetric observables.

 With a view towards the discussion of instantons in section \ref{sec_Datis}, we modify the 4D action by adding a $\theta$-term
\[
	 S_{YM} \rightarrow S_{YM} - \frac{i\theta}{8\pi^2}\Tr\int F \wedge F~ ,
\]
where $\theta$ is constant.
This term is supersymmetric by itself.
It is now natural to define the position dependent complex coupling $\tau$ as
\[
	\tau(x) = \frac{4\pi  i }{\gYMt ( x ) } + \frac{\theta}{2\pi}~,
\]
which takes values in the upper half complex plane.
This is what will appear in the instanton partition function.

\smallskip

\subsection{The hyper-multiplets}
The 5D hyper-multiplet consists of an $SU(2)_R$-doublet of complex scalars $q^A_I,~~I=1,2$ and an $SU(2)_R$-singlet fermion $\psi^A$, with the reality conditions ($A=1,2,\cdots,2N$)
\bea
 (q^A_I)^*=\Omega_{AB}\epsilon^{IJ}q^B_J~,~~(\psi^A)^*= \Omega_{AB}C\psi^B~,\nn
\eea
where $\Omega_{AB}$ is the invariant tensor of $USp(2N)$ and $C$ is the charge conjugation matrix.

Suppressing the gauge group index, the on-shell supersymmetry variations are
\bea
\label{susyhyp}
&&\delta q_I=-2i\xi_I\psi~,\nn\\
&&\delta\psi=\Gc^m\xi_I(D_mq^I)+i\sigma \xi_Iq^I-3 t^{IJ}\xi_Iq_J~.
\eea
The 5D supersymmetric action reads
\bea
&&S_{hyp}=\int_M\textrm{Vol}_M\,\big(\epsilon^{IJ}\Omega_{AB}D_mq_I^A  D^m q_J^{B} -\epsilon^{IJ}q_I^A \gs_{AC} \gs^C_{~B}  q_J^B +
\frac{15}{2} \epsilon^{IJ}\Omega_{AB}t^2  q_I^A  q_J^B
 \nn\\
&&\hspace{1cm}-2i \Omega_{AB}\psi^A\slashed{D}\psi^B-2\psi^A\gs_{AB} \psi^B -4\Go_{AB}\psi^A\gl_Iq^{IB}-iq_I^AD_{AB}^{IJ}q_J^B\big)~.\nn\eea
We refer the reader to \cite{Hosomichi:2012ek} for more details on the hyper-multiplet.
We do not explicitly present the reduction to 4D of this action and the supersymmetry transformation rules \eqref{susyhyp}. These can be performed along the same lines as for the vector multiplet.
It is important to note that the hypermultiplet action is Q-exact \cite{Kallen:2012cs,Hosomichi:2012ek,Kallen:2012va}.

\smallskip

\subsection{Deformations of the action}\label{sec_Defo}

Here we will study supersymmetric deformations of the action \eqref{action_vector_4D} which give rise to a four dimensional theory with coupling constant $\gYM$ which is position independent. To accomplish this it is convenient to go back to the  five dimensional action \eqref{action_vector} and rewrite it in terms of cohomological (twisted) variables that make the action of supersymmetry more transparent.

The cohomological complex for Yang-Mills theory on a Sasaki-Einstein manifold was introduced in \cite{Kallen:2012cs} (see \cite{Baulieu:1997nj} for earlier work). Its bosonic variables comprise, besides the fields $A_\mu$ and $\sigma$, a two form $H_{\mu\nu}$ while the gauginos are embedded in a one form $\Psi_\mu$ and a two form $\chi_{\mu\nu}$. Appendix \ref{app:cohomological} includes a brief review of the definitions of these variables, their salient properties, and their transformation under supersymmetry.

In terms of twisted variables the supersymmetric action \eqref{action_vector} can be written as the sum of a $Q$-closed contribution and various $Q$-exact terms:
\be
\label{symact}
        S_{YM} = \frac{1}{(\gYMf)^2}\left [  CS_{3,2} ( A + \sigma \kappa ) + i\Tr \int \kappa\wedge d\kappa \wedge \Psi\wedge\Psi\right] + Q W_{vec}~,
\ee
where
\be
\begin{split}
\label{Qexactdef}
        &CS_{3,2}( A ) =\Tr \int \kappa \wedge  F \wedge F, \\
        &W_{vec} = \frac{1}{(\gYMf)^2}\Tr \int \left [ \Psi \wedge \star (-\iota_{\reeb } F - d_A\sigma ) - \frac 1 2 \chi \wedge \star H + 2 \chi \wedge \star F + \kappa\wedge d\kappa \wedge (\sigma \chi ) \right ] ~.
\end{split}
\ee
Here ${\reeb}$ is the Reeb vector and $\kappa$ its dual one-form $\kappa= g {\reeb}$. Note in particular that $\iota_{\reeb }\kappa=1$.

Supersymmetry requires the overall coefficient of the term in square brackets in \eqref{symact} to be constant. Reducing  to four dimensions along $X=\partial_\alpha$ these terms give rise to a non constant $\theta$ term proportional to $\iota_{X} \kappa$
\be
\label{nctheta}
\Tr\int_B {2\pi r\over (\gYMf)^2} (\iota_{X} \kappa) F\wedge F+ \ldots
\ee
Here the dots stand for several other terms, that are necessary to preserve supersymmetry and go away in the flat space limit $r\rightarrow \infty$. When considering the reduction of the complete action \eqref{symact} this non constant $\theta$ term is cancelled by a contribution coming from the $Q$-exact terms.

Because the supercharge squares to a translation along ${\reeb}$ we can multiply each of the $Q$-exact terms by arbitrary functions, constant along ${\reeb}$, preserving supersymmetry. Indeed such deformations are not just supersymmetric but Q-exact, thus they do not affect the value of supersymmetric observables.  In the previous subsection we have considered reducing along a $U(1)$ fiber with length given by $e^{\phi}$. This length is invariant along ${\reeb}$ hence we can multiply all the $Q$-exact terms by $e^{-\phi}$. Upon reduction we obtain a theory on the base $B$ with constant $g_{YM}$. We must however pay attention to the fact that we could not rescale the supersymmetrized $CS_{3,2}$ term in the action. As a consequence the four dimensional theory will include the non-constant $\theta$ term \eqref{nctheta} as it is no longer cancelled by the $Q$-exact terms.

\subsubsection{Weyl rescaling}

Performing a dimensional reduction along a $U(1)$ fiber of varying length leads to a four dimensional field theory with varying coupling constant. It is natural to consider if it is possible to write a supersymmetric theory on $M$ with a deformed metric, so that the fiber is of constant length. For instance, this can be  achieved via a Weyl rescaling of the five dimensional metric. Note that the factor $e^\phi$ in the metric \eqref{eq:met5D} which controls the lenght of the fiber is constant along the Reeb vector, hence it is plausible that supersymmetry can be preserved under such a rescaling.

In order to address this question we can make use of the general framework for constructing supersymmetric field theories in curved space by taking a rigid limit of supergravity coupled to matter \cite{Festuccia:2011ws}. For ${\cal N}=1$ field theories in five dimensions the appropriate rigid limit of supergravity has been studied in \cite{Pan:2013uoa, Imamura:2014ima, Pan:2015nba, Pini:2015xha}. In order to preserve supersymmetry with the Weyl rescaled metric there must exist a solution to a Killing spinor equation generalizing~\eqref{KSEse}
\be
\label{genKSE}
        D_m \xi_I - t_I^{\ J} \Gamma_m \xi_J - \mathcal{F}_{mn} \Gamma^{n}\xi_I - \frac 1 2 \mathcal{V}^{pq}\Gamma_{mpq}\xi_I = 0\,.
\ee
Here $D_m$ includes a background $SU(2)_R$ connection and $ \mathcal{F}_{mn},~ \mathcal{V}^{pq}$ are background supergravity fields. For instance $ \mathcal{F}_{mn}$ is the graviphoton field strength\footnote{An equation stemming from setting to zero the variation of the dilatino needs to be satisfied as well.}.

 We also require that the supercharge continues to square to translations along the Reeb vector ${\reeb}$ and that the solution of the generalized Killing equation \eqref{genKSE} is continuously connected to the original solution of \eqref{KSEse} as the rescaling factor approaches unity. Under these conditions the rigid limit of supergravity will provide a deformation of the theory given by \eqref{action_vector} that is supersymmetric on the Weyl rescaled manifold. This theory is given by ${\cal N}=1$ SYM minimally coupled to the Weyl rescaled metric together with terms that vanish in the flat space limit and are required by supersymmetry.

The analysis of the Killing spinor equation \eqref{genKSE}, showing that the Weyl rescaling can be performed preserving supersymmetry, is presented in appendix  \ref{app:weylrescaling}. The supersymmetric variations of the fields are deformed under rescaling, however the supercharge continues to square to translations along the Reeb vector ${\reeb}$. As a consequence it is possible to define appropriate twisted variables giving rise to a cohomological complex of the same form as in the Sasaki-Einstein case (see Appendix \ref{app:cohomological}). In particular the role of $\sigma$ in the complex is now played by ${\tilde \sigma}=e^{-\phi} \sigma$. We can write an action in terms of the twisted variables as before :
\be
\label{newcohact}
        S_{YM} = \frac{1}{(\gYMf)^2}\left [  -CS_{3,2} ( A + {\tilde \sigma} \kappa ) + i \int \kappa\wedge d\kappa \wedge \tilde\Psi\wedge \tilde\Psi + Q W_{vec} \right ]~.
\ee

Here $\kappa$ is the same one form as in the Sasaki-Einstein case. In particular $\iota_{\reeb} \kappa=1$.
By taking the $Q$-exact terms of the same form as in \eqref{Qexactdef} (rescaled by factors that are constant along ${\reeb}$) all the leading terms in the action stemming from rigid supergravity can be matched to \eqref{newcohact}. It follows that, if the twisted variables are held fixed under Weyl rescaling, the theory \eqref{newcohact} is a $Q$-exact deformation of the theory on the original Sasaki-Einstein manifold.  Upon reduction the action \eqref{newcohact} gives rise to ${\cal N}=2$ SYM on the Weyl rescaled base with constant coupling $g_{YM}$ because the length of the fiber is now constant. As before however, there will be a position dependent $\theta$ term stemming from the reduction of $CS_{3,2} $.

\section{Partition function}\label{sec-partition}
In this section, we use the results for the partition function of ${\cal N}=1$ theories on toric Sasaki-Einstein manifolds \cite{Qiu:2013pta,Qiu:2014oqa}, to compute the partition function for the reduced 4D theory.
This is done by discarding the contribution of non-zero KK modes along the $U(1)$ fibre.
The answer has a similar structure to the partition function for $\mathcal{N}=2$ on squashed $S^4$ \cite{Pestun:2007rz,Hama:2012bg}, in that it factorizes to a product of contributions corresponding to isolated points on the manifold.
On $S^4$, these points are the poles, while here they are the fixed points of the torus action on the 4D manifold.
Just as for $S^4$, half of them will support instantons while the other half support anti-instantons.

\subsection{Perturbative sector}

  The results stated here are given in terms of a generic Reeb vector field, for which we do not have a Sasaki-Einstein metric. Because the partition function depends only on the Reeb, and the cohomological complex (see appendix \ref{app:cohomological}) has a straightforward generalisation for generic Reeb, our result below is still valid.

From the computation of \cite{Qiu:2014oqa}, the perturbative contribution to the partition function of the 5D $\mathcal{N}=1$ vector multiplet couple to
 hypermultiplet in representation ${\underline{R}}$  reads
\bea
Z^{pert}
=\int\limits_{\FR{t}}da~e^{-\frac{8\pi^3 r^3}{(\gYMf)^2}\varrho\,\Tr[a^2]}\cdotp
\frac{{\det}_{adj}' ~  S_3^C(ia| \vec\reeb)}{\det_{\underline{R}}S_3^C(ia+im+\vec\xi\cdotp\vec\reeb/2| \vec\reeb)}~,\label{Z_pert_fin}
\eea
where $\varrho = \mathrm{Vol}(M) / \mathrm{Vol}(S^5)$, and $S_3^C$ is the generalized triple sine associated to the cone $C$ \cite{Tizzano:2014roa,Winding:2016wpw}, which is defined as
\bea
S_3^C ( x | \vec \go ) =\prod_{\vec m\in C\cap\BB{Z}^3}(\vec \go\cdotp\vec m+x)\prod_{\vec m\in C^{\circ}\cap\BB{Z}^3}(\vec \go\cdotp\vec m-x)~.\label{gen_multiple_sine}
\eea
Here $C^{\circ}$ is the interior of $C$, and $C$ is the moment map cone of $M$.
When the manifold is SE, there is a vector $\vec \xi$ such that $\vec\xi\cdotp\vec v_i=1 \ \forall i$, and so the product above can be written as
\bea
S_3^C ( x | \vec \go ) =\prod_{\vec m\in C\cap\BB{Z}^3}(\vec \go\cdotp\vec m+x)(\vec \go\cdotp\vec m+\vec\xi\cdotp\vec\go-x)~.\label{gen_multiple_sine_I}
\eea

The perturbative partition function for the reduced theory on $B$ is obtained by keeping only the zero Kaluza-Klein modes along the $S^1$ fiber.
Using the explicit description in section~\ref{sec_Acr} of the cone $C$, the normals to its faces are
\bea
[\vec v_1,\cdots,\vec v_{\tt n},\vec v_{{\tt n}-1},\cdots,\vec v_{2{\tt n}}]=\left[
\begin{array}{cccccccc}
    1& 1 & 1 & \cdots & 1 & 1 & 1 & 1\\
    0 & 1 & 2 & \cdots & {\tt n} & {\tt n}-1 & \cdots & 1\\
    0 & 0 & * & * & * & * & * & *\\
\end{array}\right].\nn
\eea
Because $\vec X=[0;0;1]$, we keep only the modes $\vec X\cdotp \vec m=m_3=0$.
Geometrically, this is the intersection of the cone $C$ with the plane with normal vector $\vec X$.
Now the constraint $\vec v_i\cdotp \vec m\geq0,\,\forall i$ reads
\bea
 m_1+pm_2\geq0~,~~~p=0,\cdots,{\tt n}~.\nn
 \eea
As a result, the region for $m_1,m_2$ is as in figure \ref{fig_region_m}.
\begin{figure}
\begin{center}
\begin{tikzpicture}[scale=1.2, fill opacity=.4, draw opacity=1, text opacity = 1]
\draw [step=0.3,thin,gray!40] (-.1,-.5) grid (1.8,1.5);

\draw [->] (-.3,0) -- (1.8,0) node [below] {\scriptsize$m_1$};
\draw [->] (0,-.4) -- (0,1.5) node [left] {\scriptsize$m_2$};
\fill[orange] (0,0) -- (0, 1.5) -- (1.8,1.5) -- (1.8, -0.5) -- cycle;
\draw [-] (0,0) -- (1.8,-0.5) node [right] {\scriptsize$m_1+{\tt n}m_2=0$};
\end{tikzpicture}
\caption{$m_1, m_2$ runs over the shaded region, which is the cone $\tilde C$.}\label{fig_region_m}
\end{center}
\end{figure}
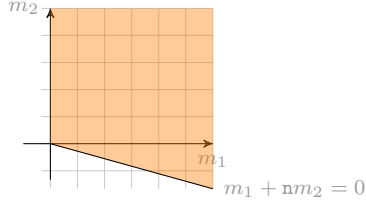
This region is a 2D cone $\tilde C$.
To describe the resulting perturbative partition function, we define the lower dimensional analogue of~$S_3^C$:
\bea
	\Gu^{\tilde C}( x | \vec \go ) = \prod_{\vec m \in \tilde C\cap\mathbb{Z}^2 } ( \vec \go \cdot \vec m + x )  \prod_{\vec m \in \tilde C^\circ \cap\mathbb{Z}^2 } ( \vec \go \cdot \vec m - x )~.
\eea
This special function is the even-dimensional analogue of the multiple sine functions that appear in odd dimensions. It is a straightforward generalization of the perturbative answer that appeared in \cite{Pestun:2007rz},
 see also \cite{Pestun:2016jze}\footnote{It is possible to write an expression for $\Gu^C$ as a factorized product over contributions from the torus fixed points.
While we do not write it explicitly here, this factorization will be apparent in the next subsection. }.

In terms of the above special function, the perturbative result can be written as
\bea
Z^{pert}
=\int\limits_{\FR{t}}da~e^{ - \frac{8\pi^3 r^3}{(\gYMf)^2 }\varrho\,\Tr[a^2]}\cdotp
\frac{{\det}_{adj}' ~  \Gu^C (ia|\reeb^1,\reeb^2)} {\det_{\underline{R}}\Gu^C ( ia+im+\vec\xi\cdotp\vec\reeb/2| \reeb^1,\reeb^2) }~.\label{Z_pert4}
\eea
Notice that it is the 5D YM coupling that appear in the classical action.
In the next section, we will explain how the combination ${\varrho \over (\gYMf)^2 }$ actually is a natural 4D quantity that involves the position dependent 4D coupling evaluated at the torus fixed points.

Next, we will investigate the asymptotic behavior of  \eqref{Z_pert4} as we go the large radius limit where the local geometry approaches flat space and we can  compare with well-known flat space results.

For computing the asymptotic behavior, we will use the approach of \cite{Qiu:2013pta}, for details we refer the reader to section 6 of that paper.
The asymptotic behavior of the above matrix model is given by
\bea \label{eq:4dpertanswer}
	Z^{pert} \sim \int\limits_{\FR{t}}da~e^{ - \frac{8\pi^3 r^3}{(\gYMf)^2 }\varrho\,\Tr[a^2]}\cdotp
	e^{ \Tr_{adj} V_v^{asy}(ia) } \cdotp e^{\Tr_R V_{h}^{asy} (ia + im) } ~,
\eea
where the functions $V_v^{asy}, V_h^{asy}$ give the asymptotic contributions of the vector and hyper-multiplet respectively.
They are given by
\bea \label{eq:hyperasymp}
	V_h^{asy} (x) = -  \rho \big [6x^2 - (4x^2 + \frac{1}{3} ( - \go_1^2 + 2 {\tt n} \go_1 \go_2 + 2 \go_2^2 ) )\log |x| \big ]~ ,
\eea
and
\bea \label{eq:vectorasymp}
	V_v^{asy} (x) = \rho \big [ 6 x^2    - ( 4 x^2  - \frac 2 3 ( \go_1^2 + {\tt n} \go_1 \go_2 - \go_2^2 ) ) \log | x | \big ]~ .
\eea
Here $\rho$ is given by
\be
\label{defrho}
	\rho = \frac{{\tt n}}{4\reeb_2 ({\tt n}\reeb_1 - \reeb_2 ) } ~.
\ee
We note that this only depends on the number of sides of our moment map cones, i.e. the number of fixed points, and is independent of the overall shape of the cone.

\subsection{Comparison with flat space results}\label{sec:flatspace}
We consider the asymptotic contributions from the vector and hyper as given above in equations \eqref{eq:hyperasymp}, \eqref{eq:vectorasymp}.
The two terms contribute to the effective action at the point $\sigma=i a$ on the Coulomb branch. We focus on the $a^2\log (r|a|)$ terms and compare them to the 1-loop $\gb$-function of $\mathcal{N}=2$ SYM in flat space.
Putting together the classical action and the quantum generated effective action, focusing only on the log term we have
\bea
&& S_{\textrm{eff}}=-\frac{8r^3}{(\gYMf)^2}\textrm{Vol}_M\,\Tr_f[ a^2]-\frac{{\tt n}r^2}{\reeb_2(\reeb_2-{\tt n}\reeb_1)}\log ( r|a|) \Tr_{adj}[a^2] \nonumber \\
&&+\frac{{\tt n}N_fr^2}{\reeb_2(\reeb_2-{\tt n}\reeb_1)}\log (r |a|)\Tr_R[a^2]~,\label{S_eff}\eea
where in the log we have  $r^{-1}$ as the renormalisation scale at which $\gYM$ is defined, and we have extracted the powers of $r$ from the volume, so $\textrm{Vol}_M$ here is just a number.

We now rewrite the above as a sum of contributions from the fixed points of $U(1)^2$ acting on $B$, which are also the loci of the closed Reeb orbits.
To this end we rewrite the volume $\textrm{Vol}_M$ as
\bea
 \textrm{Vol}_M=\pi^3\sum_i\frac{-[X,v_i,v_{i+1}]^2}{[\reeb,v_i,v_{i+1}][\reeb,v_i,X][\reeb,v_{i+1},X]}~.\label{volume_alt}\eea
This formula is derived using localisation techniques on K-contact manifolds in \cite{contact_loc}.
The sum is over the corners of $\Gd_{\mu}$ (see section \ref{sec_Acr} for notations).
At each of these corners resides a closed Reeb orbit, and each contributes ($v_i=u,\,v_{i+1}=v$)
\bea {\pi^2\over 2} \ell_O\cdotp\frac{1}{\ep_1\ep_2}\cdotp (\iota_X\gk)^2=  {\pi^2\over 2} \frac{2\pi}{[u,v,\reeb]}\cdotp\frac{-[\reeb,u,v]^2}{[\reeb,u,X][\reeb,v,X]}\cdotp\big(\frac{[X,u,v]}{[\reeb,u,v]}\big)^2\,,\nn\eea
where $\ell_O$ is the length of the closed Reeb orbit $O$ as measured by the contact form $\gk$; $\ep_{1,2}$ are the weights of $X$ acting on the space transverse to $O$ \footnote{ To match with \cite{contact_loc} it is useful to note that $(\iota_X\gk)^2$ is the 0-form component of the equivariantly completed form $(d\gk)^2$.}.
It is an interesting exercise to show that the sum in \eqref{volume_alt} actually is independent of $X$; as of course the volume of $M$ should be.
Note also that for certain choices of $\reeb,\,X$, summands of \eqref{volume_alt} may be ill-defined, yet the total sum still makes sense.

On the other hand (note $[X,v_i,v_{i+1}]=\pm1$)
\bea &&\sum_i\frac{-[\reeb,v_i,v_{i+1}]^2}{[\reeb,v_i,X][\reeb,v_{i+1},X]}\cdotp\big(\frac{[X,v_i,v_{i+1}]}{[\reeb,v_i,v_{i+1}]}\big)^2
=\sum_i\frac{-1}{[\reeb,v_i,X][\reeb,v_{i+1},X]}\nn\\
&=&\frac{-1}{(-\reeb_2)(\reeb_1-\reeb_2)}+\frac{-1}{(\reeb_1-\reeb_2)(2\reeb_1-\reeb_2)}+\cdots
+\frac{-1}{(({\tt n}-1)\reeb_1-\reeb_2)({\tt n}\reeb_1-\reeb_2)}\nn\\
&&+\frac{-1}{({\tt n}\reeb_1-\reeb_2)(({\tt n}-1)\reeb_1-\reeb_2)}+\cdots+\frac{-1}{(\reeb_1-\reeb_2)(-\reeb_2)}\nn\\
&=&\frac{2{\tt n}}{\reeb_2({\tt n}\reeb_1-\reeb_2)} = 8 \rho~, \nn\eea
where we recognize  $\rho$ defined in  \eqref{defrho}.
Using these two results, we can write both  $\rho$ and $\mathrm{Vol}_M$ as a sum over the corners of $\Delta_\mu$.
Starting from \eqref{S_eff} the $i^{th}$ corner contribution is,
\bea
S_{\rm eff}\big|_{i^{th}\,\rm{corner}}\stackrel{\log}{\sim}\Big(-\frac{4\pi^2 }{(\gYMf)^2}\frac{2\pi r }{[\reeb,v_i,v_{i+1}]}-\frac{1}{2}\log (r |a| )\frac{c_{adj}}{c_f}
+\frac{N_f}{2}\log (r |a|)\frac{c_R}{c_f}\Big)\frac{r^2}{[\reeb,v_i,X][v_{i+1},\reeb,X]}\,\Tr_f[a^2],\nn\eea
where $c_R$ is the Casimir in the representation $R$, i.e. $\Tr_R[t^at^b]=c_R\gd^{ab}$.
To interpret this formula, we note that the local geometry close to a corner is that of $S^1\times\BB{C}^2$ where the radius of $S^1$ is
\bea \frac{r}{[\reeb,v_i,v_{i+1}]} = r e^{\phi(x_i)}~.\nn\eea
So we can recognize the 4D position dependent coupling constant,
\be
\frac{2\pi r e^{\phi(x_i)} } { (\gYMf)^2 } = \frac{1}{\gYMt ( x_i )}~,
\ee
as defined in \eqref{eq:4DgYM}.
Furthermore,
$[\reeb,v_i,X]$, $[v_{i+1},\reeb,X]$ are proportional to the weights of $X$ acting on the space transverse to $S^1$, i.e. $\BB{C}^2$. It follows that
\bea \frac{1}{[\reeb,v_i,X][v_{i+1},\reeb,X]}=\textrm{Vol}_{\BB{C}_{\ep_1,\ep_2}}~,\nn\eea
i.e. the volume of $\BB{C}^2$ computed equivariantly. Putting these together
\bea && S_{\rm eff}\big|_{i^{th}\,\rm{corner}}\stackrel{\log}{\sim}\Big(-\frac{4\pi^2}{ {g}_{YM}^2(x_i)}-\frac{1}{2}\log (r |a| )\frac{c_{adj}}{c_f}
+\frac{N_f}{2}\log (r |a| )\frac{c_R}{c_f}\Big)\textrm{Vol}_{\BB{C}_{\ep_1,\ep_2}}\Tr_f[r^2 a^2]~. \nn
\eea
The quantity in the brace gives the well-known 1-loop running coupling for $\mathcal{N}=2$ theories.

The formula \eqref{volume_alt} for $\textrm{Vol}_M$ can be used to write the classical action as
\[
	S_{cl} (ia) = \sum_{i} \frac{4\pi^2 r^2} { \gYMt (x_i) } \frac{1}{\epsilon_1^i \epsilon_2^i } \Tr [ a^2 ] ~ .
\]
This matches known results on the squashed $S^4$ \cite{Pestun:2007rz,Hama:2012bg}.

\subsection{Instanton sector}\label{sec_Datis}

For the instanton sector we proceed with the same strategy as for the perturbative sector. We restrict the 5D results to the zero KK mode along $X$.
The instanton sector for the 5D theory is computed by gluing together flat space results
\bea \label{eq:5dinst}
Z_{\mathrm{inst}}^{\BB{C}^2\times S^1}(a|\gb,\ep_1,\ep_2)~,
\eea
one copy for each closed Reeb orbit. Here $\beta$ is the radius of the Reeb orbit, and $\ep_1,\ep_2$ are the equivariant rotation parameters, which are determined by the local geometry \cite{Qiu:2014oqa}. The role of instanton counting parameter is played by $q = \exp [-16\pi^3 \frac{\beta}{(\gYMf)^2} ]$.
The argument leading to this result is that the point like instantons propagating along closed Reeb orbits are the only solution invariant under the torus action. A rigorous proof is not available at the moment, though in 4D and $\BB{P}^2$ some tests have confirmed this expectation \cite{Bershtein:2015xfa}.

In 5D, the instanton equation reads
\bea \label{eq:5dinstanton}
F_5=-*_5\gk \wedge F_5~,
\eea
which is called the contact instanton equation.
After reducing to 4D, it turns into some PDE's, whose expression at a general point is not illuminating. They can however be analyzed at the fixed-points of the torus action, which correspond to the corners of the polygon~$\Gd_{\mu}$\footnote{We emphasize that some of the 4-manifolds $B$ are not toric K\"ahler, we are merely using the moment polygon of $M$ to visualize the geometry of $B$.}.

Take the hexagon example \eqref{ex_hex}.
In figure \ref{fig_hexagon_ore} we have marked with $\pm$ whether $X$ aligns or anti-aligns with $\reeb$, (e.g. at the corner 61, since $[\vec v_6,\vec v_1,\vec X]=-1$ we get anti-alignment.)
\begin{figure}[h]
\begin{center}
\begin{tikzpicture}[scale=0.5]
\draw[-,blue] (0,0) node[left] {\scriptsize{$-$}} -- (1.7,-1) node[below] {\scriptsize{$+$}} -- (3.4,0) node[right] {\scriptsize{$+$}} -- (3.4,2) node[right] {\scriptsize{$+$}} -- (1.7,3) node[above] {\scriptsize{$-$}} -- (0,2) node[left] {\scriptsize{$-$}} -- (0,0);

\node at (.7,-.7) {\scriptsize1}; \node at (2.7,-.7) {\scriptsize2}; \node at (3.6,1) {\scriptsize3};
\node at (2.7,2.7) {\scriptsize4}; \node at (0.7,2.7) {\scriptsize5}; \node at (-.2,1) {\scriptsize6};
\end{tikzpicture}\caption{The momentum polygon $\Gd_{\mu}$ whose normals are \eqref{ex_hex}. }\label{fig_hexagon_ore}
\end{center}
\end{figure}
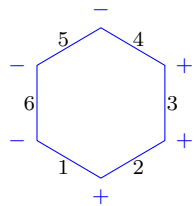
At any toric fixed point where two of the torus actions degenerate, we can decompose the vector $X$ as its part along $\reeb$ and the rest:
\[
X = (\iota_X \kappa) \reeb + X^{\perp},
\]
where $X^{\perp}$ is a locally degenerate vector with zero norm at the fixed point. Using that the Reeb is normalized,
we have that at the fixed point,
\[
\langle X,X\rangle  = (\iota_X \kappa)^2 \langle \reeb,\reeb\rangle = (\iota_X \kappa)^2~,
\]
on the other hand from the metric \eqref{eq:met5D} we have $\langle X,X\rangle = e^{2\phi}$ everywhere (here we set $r=1$). Dropping the zero norm part we conclude that at the fixed point
\be
	X = (\iota_X \kappa) \reeb  = \pm e^{\phi} \reeb~.
\ee

With this observation, we can reduce the 5D instanton equation at a fixed point as follows.
The horizontal part of the 5D instanton equation \eqref{eq:5dinstanton} reads
\[
	F_4 + \tilde \varphi d\beta = - \star_5 \kappa \wedge ( F_4 + \tilde \varphi d\beta) = - \iota_{\reeb} \star_5 (F_4 + \tilde \varphi d\beta)
	=  \pm e^{-\phi} \iota_{  X } \star_5 (F_4 + \tilde \varphi d\beta ) = \pm \star_4 ( F_4 + \tilde \varphi d\beta)~ ,
\]
where $\star_4$ in the final step uses the metric of the 4D base. The factor $e^{-\phi}$ precisely makes up the difference between the 5D and 4D volume form.
The 5D contact instanton equations at the fixed-points thus reduce to the deformed 4D instanton equations
\be \label{eq:4Dinst}
	F = \pm \star F - \tilde \varphi ( d\beta \mp \star d\beta)~ ,
\ee

which we also can write as
\[
	F^{\pm} = -\tilde \varphi d\beta^{\pm} ~.
\]

At each fixed point, keeping only the zero KK mode of the 5D answer \eqref{eq:5dinst}, we get
\be
\prod_i Z_{\mathrm{inst}}^{\mathbb{C}^2} ( a, q_i | \epsilon_1^i,\epsilon_2^i )\ ,
\ee
where $q_i = \exp [ -\frac{8\pi^2}{\gYMt (x_i) } ]$.
This is valid before turning on a $\theta$-term. With $\theta\neq 0$ in 4D, we need to distinguish which of these contributions will arise from point-like instantons versus anti-instantons.
This can be seen from the local behavior of the reduced contact instanton equation \eqref{eq:4Dinst}, namely we will have an instanton wherever $X$ and $\reeb$ align, and an anti-instanton when they anti-align.

Thus the counting parameter for an instanton at fixed point $i$ will be
\[
	e^{  2\pi i \tau (x_i)  } = q_i \ .
\]
Similarly the counting parameter for an anti-instanton will be $\bar q_j=\exp ( - 2\pi i \bar \tau(x_j)  )$.

Using the form of the normals given in \eqref{eq:normalscanonical}, the total instanton partition function will be a product of $2\mathtt{n}$ factors, each of which depends holomorphicaly on either $q$ or $\bar q$,
\[
	Z_{\mathrm{inst} }^{B} ( a | \vec \reeb )  = \prod_{i=1}^{\mathtt{n} } Z_{\mathrm{inst}}^{\mathbb{C}^2 } ( a , q_i | \epsilon_1^i, \epsilon_2^i ) \times \prod_{i=\mathtt{n}+1}^{2 \mathtt{n} } Z_{\mathrm{inst}}^{\mathbb{C}^2 } ( a , \bar q_i | \epsilon_1^i, \epsilon_2^i )~.
\]
The fact that the first $\mathtt{n}$ fixed points support instantons and the rest anti-instantons follow from the form of the normals. Hence for the geometries we are considering, by construction we will always have an equal number of instanton and anti-instanton contributions.

We can further  write down the equivariant parameters at fixed-point $i$. For $i=1, \ldots,\mathtt{n}$ these are
\[
	\epsilon^i_1 = (i-1)\reeb_1 - \reeb_2 ~, \ \ \ \ \epsilon^i_2 = i \reeb_1 - \reeb_2 ~,
\]
and for $i=\mathtt{n}+1,\ldots, 2\mathtt{n}$ they are
\[
	\epsilon^i_1 = ( 2\mathtt{n} + 1 - i)\reeb_1 - \reeb_2~, \ \ \ \ \ \epsilon^i_2 = (2\mathtt{n}-i)\reeb_1 - \reeb_2~.
\]
Thus for a given $i\leq\mathtt{n}$, fixed points $i$ and $2\mathtt{n}+1-i$ have the same equivariant parameters, and they will support instantons and  anti-instantons respectively.
It is thus tempting to try and combine the two corresponding instanton partition functions into something of the form $ | Z_{\mathrm{inst}} ( q_i ) |^2 $, however this cannot be done in general because of the position dependent $\tau$.
In other words, generically we do not have $\tau(x_i) = \tau(x_{2\mathtt{n}+1-i})$, since $\tau(x_i)$ unlike $\epsilon^i_{1,2}$  depends on the shape of the polygon.
However, for all  $\mathtt{n}$, there exists a sub-class of polytopes, that allow a choice of Reeb  for which $\tau(x_i) = \tau(x_{2\mathtt{n}+1-i})$.
In particular, inside this class are polytopes that have $\mathbb{Z}_2$ symmetry about their diagonal axis, see figure \ref{fig:symmetricpolytope}.
In these cases, after appropriately selecting $\reeb$, the instanton partition function takes the form
\[
	Z_{\mathrm{inst} }^{B} ( a | \vec \reeb )  = \prod_{i=1}^{\mathtt{n} } | Z_{\mathrm{inst}}^{\mathbb{C}^2 } ( a , q_i | \epsilon_1^i, \epsilon_2^i ) |^2~ ,
\]
which very closely mimics the answer on $S^4$ found in \cite{Pestun:2007rz}.

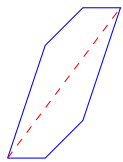
\begin{figure}[h]
\begin{center}
\begin{tikzpicture}[scale=0.5]
\draw[-,blue] (0,0) -- ( 1, 0 ) -- ( 2,1) -- ( 3, 4 )  -- ( 2 ,4  ) -- ( 1, 3) --  ( 0 , 0 ) ;
\draw[dashed,red] ( 0,0) -- ( 3,4);
\end{tikzpicture}
\caption{An example of the polygons that exhibit a reflection symmetry about the dashed red line. For the corresponding manifold the local geometry at the fixed points corresponding to parallel sides will be the same, and as a consequence the instanton partition function can be written as something explicitly real.  }\label{fig:symmetricpolytope}
\end{center}
\end{figure}

For concreteness we will work out one example in some detail. Consider again the $S^2\times S^2$ coming from the reduction of a $Y^{p,q}$ space.

\begin{example} [$S^2 \times S^2$, from reduction of $Y^{p,q}$] \label{example:S2xS2inst}
 For these spaces with topology $S^2~\times~S^2$ (we emphasize that their metric is not canonical), we have 4 fixed-points, corresponding to the ways of choosing one pole from each of the $S^2$'s.
Using the normals given in \eqref{eq:Ypqnormals}, and the prescriptions explained in section \ref{sec:flatspace} for computing the local data, we find the following parameters
 \begin{center}
  \begin{tabular}{ | c | c | c | c | c | }
    \hline
     & 1 & 2 & 3 & 4 \\ \hline
     $[v_i,v_{i+1},X]$ & +1 & +1 & -1 & -1 \\ \hline
    $\beta^{-1} = e^{-\phi(x_i)} $ & $\reeb_3$ & $\reeb_3 + (p-q)(\reeb_1-\reeb_2) $ & $ -\reeb_3 + p\reeb_1 + q (\reeb_1-\reeb_2) $ & $- \reeb_3 + p \reeb_2$ \\ \hline
    $\epsilon_1$ & $-\reeb_2$ & $\reeb_1-\reeb_2$ & $2\reeb_1-\reeb_2$  & $ \reeb_1-\reeb_2$  \\ \hline
  $\epsilon_2$ & $\reeb_1-\reeb_2$ & $2\reeb_1-\reeb_2$ & $\reeb_1-\reeb_2$ & $-\reeb_2$ \\ \hline
    \end{tabular}
\end{center}
The first line of the table tells us if $X$ and $\reeb$ align or anti-align at the fixed-point, and the second line tells us the inverse radius of the Reeb orbit.
The local complex coupling $\tau(x_i)$ at fixed-point $i$ given by
\[
	\tau (x_i) = \frac{4\pi i \beta_i}{g_{YM}^2} + \frac{\theta}{2\pi}~ .
\]
We see that the expression for $\beta_i$ depends on $\reeb_3$ and on the shape of the polygon, i.e. on the integers $p,q$, whereas the equivariant parameters $\epsilon_{1,2}$ do not.
Thus the instanton partition function for these spaces takes the form
\[\begin{split}
	&Z^{\mathbb{C}^2}_{\mathrm{inst} } ( a |  q_1 , -\reeb_2, \reeb_1 - \reeb_2 )
	Z^{\mathbb{C}^2}_{\mathrm{inst} } ( a | q_2 , \reeb_1 - \reeb_2, 2\reeb_1 - \reeb_2 ) \\&	\times
	Z^{\mathbb{C}^2}_{\mathrm{inst} } ( a |  \bar q_3 , 2\reeb_1 - \reeb_2 , \reeb_1 - \reeb_2  )
	Z^{\mathbb{C}^2}_{\mathrm{inst} } ( a |  \bar q_4 , \reeb_1 - \reeb_2, -\reeb_2 )  ~.
\end{split}
\]

\end{example}

\section{Summary}

In this paper we have constructed $\mathcal{N}=2$ 4D gauge theories on a wide class of toric manifolds with the topology of
 connected sums of $S^2 \times S^2$. The construction comes from the reduction of toric Sasaki-Einstein manifolds along a free $U(1)$
 chosen in such a way that it preserves 5D Killing spinors. We can reduce to the 4D geometrical data from the 5D toric Sasaki-Einstein
   geometry. However at the moment we are missing a description of the 4D geometry in intrinsic 4D terms. It would be important to
     further study this 4D geometry and see if our 4D toric manifolds are part of a bigger class of toric manifolds which allow $\mathcal{N}=2$ theories.
     Another important issue is that the resulting 4D theories have a point dependent coupling constant.  We think that this feature of the 4D theory should be
      taken seriously, and one should study the supergravity origin of this when placing the supersymmetric theory on curved manifolds.

  We also calculated the exact partition function for these 4D theories, by reducing the 5D answer.
  The toric manifolds we consider have an even number of fixed points, half of which
   corresponds to instanton contribution to the partition function and half to anti-instanton contributions. Although $S^4$ is formally outside of
    our analysis, the formal expression for the partition function coincides with Pestun's result (as well as for the squashed $S^4$).  We conjecture that the $\mathcal{N}=2$ partition function on any toric simply connected 4D manifold will have the same structure that we have found here.
     It would be curious to see if our results has any AGT-like explanation coming from the reduction of 6D theories, especially taking into account
      the fact that the coupling $\tau$ is point dependent.
      
      In this paper we concentrate on the reduction of 5D supersymmetric gauge theory to 4D supersymmetric gauge theory and our framework requires that the underlying four manifold is spin. The present analysis can be extended to a wider class of theories and manifolds for instance by formulating it in cohomological terms. Namely we can start from 5D cohomological theory defined in   \cite{Kallen:2012cs}  and reduce down to 4D cohomological theory. In this case we should expect a similar effect when in 4D we glue together the contributions of instantons for some of the fixed points and contributions of anti-instantons for the other fixed points. We plan to explore these more general 4D theories elsewhere.

\bigskip
{\bf Acknowledgements:} We thank first of all Dario Martelli who shared with us his insight that inspired this project. 
We also thank Marco Chiodaroli, Marco Golla, Thomas Kragh, Yiwen Pan and Luigi Tizzano for helpful discussions, and finally Mark Hamilton, who pointed out an imprecision in our statement about 4-manioflds. The work of MZ  is supported in part by Vetenskapsr\r{a}det
 under grant \#2014-5517, by the STINT grant, and by the grant  ``Geometry and Physics"  from the Knut and Alice Wallenberg foundation.
The work of GF is supported by the ERC STG grant 639220.

\appendix

\section{Details of the reduction conditions} \label{app_proofs}
Here, we give the proofs of statements in section \ref{sec_pdab}.
Remember that $S^1\to M \to B$ is our nontrivial circle bundle and that $E \to M$ is the bundle we wish to push down.
\begin{proposition}\label{prop_criterion_II_app}
  We use $\ga$ as the coordinate of the circle fibre and we let $A$ be a connection of $E$, then if $P\exp i\int_0^{2\pi} d\ga A_{\ga}=id$, the bundle $E$ can be pushed down.
\end{proposition}
\begin{proof}
 The bundle $M$ can be trivialised as $S^1\times U_i$, where $\{U_i\}$ is a cover of $B$. Then choose a cover of $M$ of the form $V_{is}=(a_{is},b_{is})\times U_i$, where $(a_{is},b_{is}),~s=1,\cdots n_i$ covers an interval of the circle fibre. Assume that the cover is chosen fine enough so that $E$ is trivialised over $\{V_{is}\}$, and let $g_{is,jt}$ be the transition function of $E$. We first show that the transition function can be made independent of $\ga$.

  On a patch $V_{is}$, we denote the connection as $A^{is}$ and so on the intersection $V_{is}\cap V_{jt}$ the connections are related as
  \bea
   A^{is}=g_{jt,is}^{-1}dg_{jt,is}+g_{jt,is}^{-1}A^{jt}g_{jt,is}~.\nn
  \eea
  We first adjust locally the trivialisation by multiplying with a Wilson line
  \bea
   h_{is}(\ga,x)=P\exp\int_{a_{is}}^{\ga}d\ga\, (-A^{is}_{\ga})~,\nn
   \eea
  where $(\ga,x)$ are the fibre and base coordinates. The Wilson lines $h_{is}$ satisfies
  \bea \partial_{\ga}h_{is}(\ga,x)=-A_{\ga}^{is}h_{is}(\ga,x)~.\nn\eea
  Then the new transition function becomes
  $\tilde g_{is,jt}=h_{is}^{-1}g_{is,jt}h_{jt}$. Let us look at
  \bea
  \partial_{\ga}\tilde g_{is,jt}&=&\partial_{\ga}(h_{is}^{-1}g_{is,jt}h_{jt})=h^{-1}_{is}A^{is}_{\ga}g_{is,jt}+h_{is}^{-1}(\partial_{\ga}g_{is,jt})h_{jt}-h_{is}^{-1}g_{is,jt}A^{jt}_{\ga}h_{jt}\nn\\
  &=&h_{is}^{-1}g_{is,jt}\Big(g_{is,jt}^{-1}A^{is}_{\ga}g_{is,jt}-A^{jt}_{\ga}+g_{is,jt}^{-1}\partial_{\ga}g_{is,jt}\Big)h_{jt}=0~.\nn\eea
  So the new transition function is independent of the $\ga$ direction, and the new connection satisfies $\tilde A_{\ga}^{is}=0$ by construction.

  Now we assume that such adjustment has been made and all transition functions are $\ga$ independent and $A_{\ga}=0$.
  Now for each fixed $i$, we readjust the trivialisation on $V_{i2},\,V_{i3},\cdots V_{in_i}$ by multiplying by
  \bea g_{i1,i2},~g_{i2,i3}g_{i1,i2},~g_{i3,i4}g_{i2,i3}g_{i1,i2},\cdots\nn\eea
  This way we can make the transition functions $g_{is,i(s+1)}$ identity except possibly when one goes a full circle, i.e. $g_{in_i,i1}$ (nonetheless it is still $\ga$-independent). But the qantity $g_{in_i,i1}$ can be computed as the holonomy along the circle fibre of the original connection $A$. If this last holonomy is trivial then $g_{in_i,i1}=1$, which means that $V_{i1},V_{i2},\cdots, V_{in_i}$ for fixed $i$ can be pieced together and become one single open set that cover s the whole circle, i.e. $V_i=[0,2\pi]\times U_i$. The transition is by construction $\ga$-independent, and so one can push down the bundle $E$ to $B$
\end{proof}

\begin{proposition}
  Suppose that the criterion in proposition \ref{prop_criterion_II_app} is satisfied, then the sections of $E$ satisfying $D_{\ga}s=0$ can be pushed down.
\end{proposition}
\begin{proof}
  Repeat the adjustment as in proposition \ref{prop_criterion_II_app} to make the transition functions independent of $\ga$, and we continue to use the notation there. Let $s$ be a section, if on each patch $V_i$ one has $\partial_{\ga}s=0$, then clearly $s$ can regarded as a section of the pushdown bundle on $B$.

  Now undo the adjustments of trivialisation then $\partial_{\ga}s=0$ reverts to $D_{\ga}s=0$.
\end{proof}

\section{Spinor conventions and bilinears}\label{app_spinors}

We follow the convention for spinors of \cite{Hosomichi:2012ek}. Let $\{e^{\tt a}\}$ be a set of vielbein which reduces the structure group of $M^5$ to $SO(5)$. The gamma matrices satisfy the Clifford algebra
\bea
\{\Gc^{\tt a},\Gc^{\tt b}\}=2\delta^{\tt ab}~,\nn
\eea
and the charge conjugation relation
\bea
C^{-1}(\Gc^{\tt a})^TC=\Gc^{\tt a}~,~~~~C^T=-C,~C^*=C~.\nn
\eea
Denote by $\Gc_m=\Gc^{\tt a}e^{\tt a}_m$, which satisfy $\{\Gc_m,\Gc_n\}=2g_{mn}$.

The spinor bi-linears are formed using $C$,
\bea
\psi^TC\chi\stackrel{\textrm{abbreviate}}{\longrightarrow}\psi\chi~,\label{spinor_bi_linear}
\eea
throughout the paper, the bi-linears are abbreviated as $(\psi\chi)$, following \cite{Hosomichi:2012ek}.

Denote by
\bea \Gc^{\tt a_1\cdots a_n}=\frac{1}{n!}\Gc^{[\tt a_1}\cdots \Gc^{\tt a_n]}\nn\eea
and similarly for their curved space counterpart. We use Dirac's slash notation
\bea \slashed{M}=M\cdotp\Gc=M_{i_1\cdots i_p}\Gc^{i_1\cdots i_p},~~M\in \Go^p(M),\nn\eea
we will even drop the slash whenever confusion is unlikely.

The $SU(2)$ R-symmetry index are raised and lowered from the left
\bea
\xi^I=\epsilon^{IJ}\xi_J~,~~~~ \xi_I=\epsilon_{IJ}\xi^J~,~~~~ \epsilon^{IK}\epsilon_{KJ}=\delta^I_J~,~~~~ \epsilon^{12}=-\epsilon_{12}=1~.\nn
\eea

In 5D, one cannot impose the Majorana condition on a spinor, but we can instead impose the symplectic Majorana condition, which for a pair of spinors $\xi_I$ reads
\be
	\overline{\xi_I^{\alpha}} = C_{\alpha\beta}\epsilon^{IJ} \xi_J^\beta .
\ee

\subsection{Spinor bilinears and some of their properties} \label{app:spinorbilinears}

Given a symplectic Majorana spinor $\xi_I$ we can construct the following spinor bilinears out of it:
\be
	s = - \xi^I \xi_I > 0 , \ \ \ \ \reeb^m = \xi^I \Gamma^m \xi_I, \ \ \ \ \Theta_{mn}^{IJ} = \xi^I \Gamma_{mn} \xi^J .
\ee
These will satisfy various relations, and we use the following when solving the Killing spinor equation:
\be
	\reeb_m \reeb^m = s^2, \ \ \ \reeb^m \Theta_{mn}^{IJ} = 0 , \ \ \ s \Theta^{IJ}_{mn} = \frac 1 2 \epsilon_{mnpqr} \reeb^p \Theta^{IJqr}.
\ee
In particular when the spinor solves the Killing spinor equation on a SE manifold, we have $s=1$ and the vector $\reeb$ is the Reeb of our contact structure.
For this case, when we choose $t_I^{\ J} = \frac i 2 (\sigma_3)_I^{\ J}$ and denote the contact one-form as $\kappa = g(\reeb)$, we have
\be
	d\kappa_{mn} = 2 i (\sigma_3)_{IJ} \Theta^{IJ}_{mn}.
\ee

\section{The example of $Y^{p,q}$}\label{app:Ypqexample}

To make our procedure a bit more concrete, let us give some details on the example of $Y^{p,q}$, where we can write an explicit metric.
We essentially take all the relevant information from Gauntlett, Martelli, Sparks and Waldram\cite{Gauntlett:2004yd}, where the metric on $Y^{p,q}$ is given as
\be \label{eq:Ypqmetric}
\begin{split}
	ds^2 = &\frac{1-y}{6} (d\theta + \sin^2 \theta d\phi^2 ) + \frac{dy^2}{ w(y) q(y) } + \frac{q(y)}{9} [ d\psi - \cos \theta d\phi ]^2 \\
	& + w(y) \left [ d\alpha + \frac{a - 2y + y^2 }{6 ( a-y^2 ) } [ d\psi - \cos \theta d\phi ] \right ]^2
\end{split}
\ee
where
\[
\begin{split}
	&w(y) = \frac{2(a-y^2 ) } { 1-y} , \\
	&q(y) = \frac{a-3y^2 + 2y^3}{ a - y^2 } .
\end{split}
\]
Here the coordinates $(\theta,\phi, y, \psi)$ describe the 4D base and $\alpha$ describes the $S^1$ fiber.
The coordinates run over the following ranges
\[
	0 \leq \theta \leq \pi, \ \ \ 0 \leq \phi \leq 2\pi, \ \ \ y_1 \leq y \leq y_2 , \ \ \ 0 \leq \psi \leq 2\pi , \ \ \  0 < \alpha < 2\pi l  ,
\]
and the constant $a$ is chosen in the range $0 < a < 1$, for which case the equation $q(y)=0$ has one negative and two positive roots. We choose $y_1$ to be the negative root and $y_2$ to be the smallest positive root.
This makes sure that the base manifold described by $(\theta,\phi,y,\psi)$ has the topology of $S^2\times S^2$, and that $w(y) > 0$ everywhere so that $\alpha$ describe a no-where degenerating $S^1$ fiber.
More precisely, as explained in detail in \cite{Gauntlett:2004yd}, to get the proper SE manifold $Y^{p,q}$, we need to pick $a$ such that $y_2 - y_1 = \frac{3q}{2p}$; which they show that you can always do for any coprime $p > q$ .
This also fixes the constant $l$ which determine the range of $\alpha$.

The metric above makes very explicit the $S^1$ fibration structure, and in our construction we dimensionally reduce along  the $\alpha$ direction; which we emphasize is not the Reeb.
The canonical Reeb vector in these coordinates is given by
\be
	\reeb = 3\frac{\partial}{\partial \psi} - \frac{1}{2} \frac{\partial}{\partial \alpha } ,
\ee
which has constant unit norm. On the other hand we see that the radius of the $S^1$ fiber is given by $\sqrt{w(y)}$, which clearly is not constant over the base manifold.

One can also see that the 4D base manifold has the topology of $S^2 \times S^2$ by looking at the metric.
First we can observe that for fixed $y$ the first term in \eqref{eq:Ypqmetric} describes the $S^2$ covered by coordinates $(\theta,\phi)$.
We also see that $y$ is running over an interval, and the last term of the first line describes a circle parametrized by $\psi$ that degenerate at the ends of the interval (since $q(y_i) = 0$); which gives the second $S^2$.
From this, we know that the base has the structure of an $S^2$  bundle over $S^2$, but it is not immediately clear that its the trivial $S^2\times S^2$ rather than some non-trivial fibration.
However this is shown in \cite{Gauntlett:2004yd}, and we won't repeat the argument here.
So the 4D base is $S^2\times S^2$, equipped with a non-standard metric given by the first line of \eqref{eq:Ypqmetric}.

Next, let us briefly explain how this is connected to the toric description in terms of a moment map cone and its normals, a story first told in \cite{Martelli:2004wu}.
As explained in section \ref{sec_Acr} the toric picture of the manifold is that of a $T^3$ fibration over a polygon where some $S^1$ fibers degenerate as we go to the edges. At the vertices, an entire $T^2$ degenerate and we have a local geometry of $\mathbb{C}^2 \times S^1$.
For $Y^{p,q}$ we know that the base polygon has 4 edges (as seen in figure \ref{fig_hexagon}).
Each edge corresponds to a particular pole of one of the $S^2$, where the rotation of that $S^2$ has a fix point and thus degenerates.
From the metric we see that the edges are given by $\{\theta = 0\}$, $\{\theta=\pi\}$, $\{y=y_1\}$ and $\{y=y_2\}$.
We can then find vectors that generate rotations, i.e. combinations of $\partial_\phi, \partial_\psi$ and $\partial_\alpha$, that are such that their norm vanishes at each of these poles.
Doing this, we find
\be
	v_1 = \partial_\phi + \partial_\psi, \ \ \
	v_2 = \partial_\psi + \frac{p-q}{2l}\partial_\alpha, \ \ \
	v_3 = -\partial_\phi + \partial_\psi, \ \ \
	v_4 = \partial_\psi - \frac{p+q}{2l} \partial_\alpha ,
\ee
which degenerate at $\theta=0,y=y_1, \theta=\pi$ and $y=y_2$ respectively. Observe that we here rescale $\partial_\alpha$ by $1/l$ so that it has a normal period.
In this computation, we use properties of the roots $y_1,y_2$ that relates them to $p,q$:
\[
	\frac{y_1-1}{3 l y_1} = p+q, \ \ \  \frac{1- y_2}{3ly_2} = {p-q} .
\]
These vectors $v_1,\ldots,v_4$ are precisely the inward normals of the moment map cone, but to relate them to the ones given in section \ref{sec_Acr}, we need to make a change of basis. Instead of using the basis $\partial_\phi,\partial_\psi,\partial_\alpha$, we should use a basis of vectors whose orbits all close. The orbits of $l^{-1}\partial_\alpha$ closes everywhere since it is a proper fibration, but that is not true of the orbits of $\partial_\phi$ and $\partial_\psi$. One suitable basis is instead given by
\[
	e_1 = \partial_\phi + \partial_\psi, \ \ \ e_2 = -\partial_\phi + \frac{p-q}{2 l } \partial_\alpha, \ \ \  e_3 = -\frac{1}{l}\partial_\alpha,
\]
where of course this choice is far from unique: any $SL_3(\mathbb{Z})$ transformation of this give us an equally good basis.
In this basis, the vectors $v_1,\ldots,v_4$ has the components
\[
	v_1 = [1,0,0], \ \ v_2 = [ 1 , 1, 0], \ \ v_3 = [1, 2 , p-q ] , \ \ v_4 = [1,1, p]  ,
\]
and we recognize the normals of the moment map cone as given in example \ref{example:Ypq}.
So we have seen explicitly how the toric description and the explicit metric is related to one-another.

\section{The Weyl rescaled background} \label{app:weylrescaling}

As discussed in section \ref{sec_Defo}, after performing a Weyl transformation where the metric is rescaled $g \rightarrow \tilde g =e^{-2\phi} g$, we wish to show that we can still have rigid supersymmetry on this new background.
The scale factor $\phi$ is invariant along both the $U(1)$ fiber, and along the Reeb, i.e. $L_X \phi = L_{\reeb} \phi = \iota_{\reeb} d\phi = 0$.

To do this we use the minimal off-shell 5D supergravity \cite{Zucker:1999ej}, and focus on the Killing spinor equation coming from requiring the supersymmetry variation of the gravitino to vanish:
\be
	D_m \xi_I - t_I^{\ J} \Gamma_m \xi_J - \mathcal{F}_{mn} \Gamma^{n}\xi_I - \frac 1 2 \mathcal{V}^{pq}\Gamma_{mpq}\xi_I = 0 ,
\ee
where $D_m$ includes the coupling to the background $SU(2)_R$ gauge field $A_{m I}^{\ \ J}$, as $D_m \xi_I = \nabla_m \xi_I - A_{m I}^{\ \ J}\xi_J$.
Here, $\mathcal{F} = d\mathcal{A}$ is the field strength of the graviphoton, $\mathcal{V}$ is a 2-form background field and $t_I^{\ J}$ is background $SU(2)_R$ triplet scalar.
We will use a $\tilde \ $ to denote new quantities after the rescaling, while non-tilded variables denotes `old' quantities.
The idea now is that we can solve this equation by turning on these various background fields so that the new Killing spinor is a rescaling of the old one. In particular, we require that the new spinor $\tilde \xi_I$ is such that the bilinear giving us the Reeb vector is unchanged, i.e.
\[
	\reeb^m = \xi^I \Gamma^m \xi_I = \tilde \xi^I \tilde \Gamma^m \tilde \xi_I,
\]
and since $\tilde \Gamma^m = \Gamma_a \tilde e_a^m$ scales like the inverse vierbein, i.e. with $e^{\phi}$, this fixes
\[
	\tilde \xi_I = e^{-\phi/2} \xi_I.
\]
Next, we can compute how the spin connection changes because of the rescaling and then use our old solution to get rid of the derivative of the spinor from the equation.
The spin connection changes as
\be
	\tilde \go^{ab}_m = \go^{ab}_m + (\partial^n \phi ) ( e_n^a e^b_m  - e^b_n e^a_m ),
\ee
and using this as well as our old equation \eqref{KSSE}, our Killing spinor equation becomes
\be \label{eq:KiSpinor2}
 -\frac 1 2 (\partial_m \phi) \tilde \xi_I - \frac 1 2 (\partial^n \phi) \Gamma_{mn} \tilde\xi_I + A_{m I}^{\ \ J} \tilde \xi_J   + (t_I^{\ J} \Gamma_m -  \tilde t_{I}^{\ J }\tilde \Gamma_m) \tilde \xi_J - \mathcal{F}_{mn}\tilde \Gamma^n \tilde \xi_I -\frac 1 2 \mathcal{V}^{n p}\tilde \Gamma_{mnp}\tilde \xi_I = 0 .
\ee
This is now an algebraic equation for $\mathcal{F},\mathcal{V}, A$ and $\tilde t_I^{\ J}$, and it is a straightforward but somewhat tedious exercise to solve it.
When solving, it is helpful to note that a symplectic Majorana spinor $\chi_I$ is completely determined by the contractions $\xi^I \Gamma_m \chi_I$ and $\chi_{(J} \xi_{I)}$. So performing these contractions of equation \eqref{eq:KiSpinor2}, we get a set of equations that only involves spinor bilinears, and using the properties we know about the various bilinears, see appendix \ref{app:spinorbilinears}, we find the following solution for our various background fields:
\be \label{eq:backgroundsol}
\begin{split}
\mathcal{F} &= d \mathcal{A}, \ \ \ \ \mathcal{A} = -\frac 1 2 (e^{-\phi} - e^{-\phi_p} ) \kappa, \\
\mathcal{V} &= \frac 1 2 e^{-\phi} d\phi \wedge \kappa , \\
\tilde t_I^{\ J } &= - \frac i 2 ( e^\phi - 2 e^{2\phi - \phi_p} ) (\sigma_3)_I^{\ J} , \\
A_{I}^{\ J} &= -i ( 1 - e^{\phi-\phi_p} ) (\sigma_3)_I^{\ J} \kappa .
\end{split}
\ee

Here, $\phi_p$ is a free constant of our solution; and $\kappa$ is the old contact 1-form.
If we choose $\phi_p$ to be the scale factor $\phi$ evaluated at some point $p$ on our manifold, then when $\phi$ is a constant scaling, we note that the background fields $\mathcal{F},\mathcal{V}$ and $A$ all vanish, and the $\tilde t_I^{\ J}$ becomes a simple scaling of the old $t_{I}^{\ J}$ field.
This shows that our solution is smoothly connected to the old SE solution.

From supergravity, we also get a second equation that we need to solve, the dilatino equation.
This involves one further background scalar field $C$, which one can solve for directly in terms of the other background fields.
Through a tedious computation, one can then check that our solution also solves this equation so that we indeed have a valid rigid supersymmetric background.

\section{Cohomological variables}\label{app:cohomological}
The 5D supersymmetry given of our vector multiplet looks like
\bea
	Q A_m &=& i\xi^I \Gamma_m \lambda_I ~, \\
	Q \sigma &=& i \xi^I \lambda_I ~, \\
	Q \lambda_I &=& - \frac 1 2 F_{mn}\Gamma^{mn}\xi_I + (D_m\sigma)\Gamma^{m}\xi_I + D_I^{\ J}\xi_J + 2 \sigma ( t_I^{\ J} \xi_J + \frac 1 2 \mathcal{F}_{mn}\Gamma^{mn}\xi_I )~ , \\
	Q D_{IJ} &=& - i (\xi_I \Gamma^m D_m \lambda_J ) + [\sigma, \xi_I \lambda_J] + i  t_I^{\ K} \xi_K \lambda_J - \frac i 2 \mathcal{V}_{mn} (\xi_I \Gamma^{mn} \lambda_J) + (I \leftrightarrow J)~,
\eea	
where the full set of background supergravity fields are included.
In the SE case, the only non-zero background field is $t_I^{\ J}$, the others ($\mathcal{F,V}$) vanish.
We can make the structure of the supersymmetry clearer by switching to cohomological variables, following for example \cite{Kallen:2012va}.
This change of variables is given by
\[
\begin{split}
	\Psi_m &= \xi_I \Gamma_m \lambda^I~ , \ \ \ \ \chi_{mn} = \xi_I \Gamma_{mn} \lambda^I +  ( \kappa_m \Psi_n - \kappa_n \Psi_m)~ , \\
	H& = Q \chi = 2 F^+_H + \Theta^{IJ} ( D_{IJ} + 2 t_{IJ} \sigma ) ~,
\end{split}
\]
where in the last line we have used our particular SE background to only keep $t_{IJ}$.
Here $\Psi$ is a fermionic one-form, and $\chi,H$ are horizontal, transversally self-dual 2-forms.
$F_H^+$ denotes the self-dual part of the horizontal part of $F$; and we see that $\chi$ and $H$ essentially becomes the auxiliary fields.
In these variables, the supersymmetry variation takes the form of the cohomological complex \cite{Baulieu:1997nj, Kallen:2012cs},
\bea \label{eq:cohomcomplex}
Q A &=&i \Psi~ , \\
Q \Psi &=& - \iota_{\reeb} F + d_A ( \sigma ) ~, \\
Q \sigma &=& - i \iota_{\reeb} \Psi~, \\
Q \chi &=& H~ , \\
Q H &=& - i L_{\reeb}^A \chi - [ \sigma, \chi ]~ .
\eea
Written in these variables it is clear that $Q^2 = - i L_{\reeb}^A + G_{\sigma}$ where $G_{\sigma}$ denotes a gauge transformation with parameter $\sigma$, and $L_{\reeb}^A$ is the gauge covariant Lie derivative along the Reeb, $L_{\reeb}^A = L_{\reeb} + G_{\iota_{\reeb} A} $.

\subsection{Weyl rescaled case}

In the Weyl rescaled background as described in appendix \ref{app:weylrescaling}, we can perform the same change of variables.
Now, the Reeb vector that appears in our supersymmetry is no longer normalized, so we have to insert its norm in the appropriate places in our change of variables. And since the background fields appear in the variation of $\chi$, the definition of $H$ will also change.
So for our new background we make the change of variables
\[
\begin{split}
	\tilde\Psi_m &= \tilde\xi_I \tilde\Gamma_m \lambda^I ~, \ \ \ \ \tilde\chi_{mn} = \tilde\xi_I \tilde\Gamma_{mn} \lambda^I +  e^{-\phi}( \kappa_m \tilde\Psi_n - \kappa_n \tilde\Psi_m) ~, \\
	\tilde H& = Q \tilde\chi = 2 e^{-\phi} F_H^+ + 2 \sigma e^{-\phi} ( e^{-\phi}-e^{\phi_p} )  {d\kappa}^+ + ( D_{IJ} + 2\sigma t_{IJ}) \tilde \Theta^{IJ} ~,
\end{split}
\]
where we have used the specific form of our background.
Computing the supersymmetry variations of our new cohomological variables, we find that it is natural to make the field redefinition
\be
	\tilde \sigma = e^{-\phi} \sigma~,
\ee
because in terms of this field, the new complex takes the form
\bea \label{eq:cohomcomplex}
Q A &=&i \tilde \Psi ~, \\
Q \tilde\Psi &=& - \iota_{\reeb} F + d_A (\tilde \sigma )~ , \\
Q\tilde \sigma &=& - i \iota_{\reeb} \Psi ~, \\
Q \tilde \chi &=&\tilde H ~ , \\
Q \tilde H &=& - i L_{\reeb}^A\tilde \chi - [\tilde \sigma,\tilde \chi ]~ ,
\eea
which has exactly the same form as the complex before the rescaling. In the computation, various cancellations between the background fields take place, and we are left with the above result.
This is to be expected, since the parameters of the square of the supersymmetry depends on the two spinor bilinears $\tilde \xi^I\tilde \Gamma^m \tilde \xi_I = \tilde \reeb^m = \reeb^m$ and $-\sigma \tilde\xi^I\tilde\xi_I~=~e^{-\phi}\sigma=\tilde\sigma$.

\providecommand{\href}[2]{#2}\begingroup\raggedright

\bibliographystyle{utphys}
\bibliography{5Dgaugetheory}{}

\endgroup

\end{document}